\documentclass[11pt,a4paper,reqno]{amsart}%
\usepackage[utf8]{inputenc}
\usepackage{mathrsfs}
\usepackage{dsfont}
\usepackage{hyperref}
\usepackage{amsmath}
\usepackage{amssymb}
\usepackage{amsthm}
\usepackage{amsfonts}
\usepackage{amstext}
\usepackage{amsopn}
\usepackage{amsxtra}
\usepackage{mathrsfs}
\usepackage{dsfont}
\usepackage{esint}
\usepackage{graphicx}
\newtheorem{theorem}{Theorem}
\newtheorem{proposition}{Proposition}

\newtheorem{lemma}{Lemma}
\newtheorem{corollary}{Corollary}

\newcommand{\ii}{\infty}
\newcommand\R{{\ensuremath {\mathbb R} }}

\newcommand\N{{\ensuremath {\mathbb N} }}

\newcommand\1{{\ensuremath {\mathds 1} }}
\renewcommand\phi{\varphi}
\newcommand{\gH}{\mathfrak{H}}
\newcommand{\gS}{\mathfrak{S}}

\newcommand{\wto}{\rightharpoonup}

\newcommand{\cR}{\mathcal{R}}

\newcommand{\cB}{\mathcal{B}}

\newcommand{\cW}{\mathcal{W}}
\newcommand{\cF}{\mathcal{F}}

\newcommand{\cD}{\mathcal{D}}

\newcommand{\cL}{\mathcal{L}}

\renewcommand{\epsilon}{\varepsilon}

\newcommand{\norm}[1]{ \left| \! \left| #1 \right| \! \right| }
\newcommand{\tr}{{\rm Tr}}
\newcommand{\sgn}{{\rm sgn}}

\renewcommand{\ge}{\geqslant}
\renewcommand{\le}{\leqslant}
\renewcommand{\geq}{\geqslant}
\renewcommand{\leq}{\leqslant}
\renewcommand{\hat}{\widehat}
\renewcommand{\tilde}{\widetilde}
\renewcommand\d[1]{{\ensuremath{\,\text{d}#1}}}

\title[Hartree equation for infinitely many particles II]{The Hartree equation for infinitely many particles. II. Dispersion
	and scattering in 2D}

\author[M. Lewin]{Mathieu LEWIN}
\address{CNRS \& Universit\'e de Cergy-Pontoise, Mathematics Department (UMR 8088), F-95000 Cergy-Pontoise, France} 
\email{mathieu.lewin@math.cnrs.fr}

\author[J. Sabin]{Julien SABIN}
\address{Universit\'e de Cergy-Pontoise, Mathematics Department (UMR 8088), F-95000 Cergy-Pontoise, France} 
\email{julien.sabin@u-cergy.fr}

\date{\today}

\allowdisplaybreaks[1] 

\begin{document}

\thanks{\copyright\,2013 by the authors. This paper may be reproduced, in its entirety, for non-commercial purposes.}

\begin{abstract}
We consider the nonlinear Hartree equation for an interacting gas containing infinitely many particles and we investigate the large-time stability of the stationary states of the form $f(-\Delta)$, describing an homogeneous Fermi gas. Under suitable assumptions on the interaction potential and on the momentum distribution $f$, we prove that the stationary state is asymptotically stable in dimension 2. More precisely, for any initial datum which is a small perturbation of $f(-\Delta)$ in a Schatten space, the system weakly converges to the stationary state for large times.
\end{abstract}

\maketitle

\tableofcontents

%%%%%%%%%%%%%%%%%%%%%%%%%%%%%%%%%%%%%%%%%%
%%%%%%%%%%%%%%%%%%%%%%%%%%%%%%%%%%%%%%%%%%
\section{Introduction}
%%%%%%%%%%%%%%%%%%%%%%%%%%%%%%%%%%%%%%%%%%
%%%%%%%%%%%%%%%%%%%%%%%%%%%%%%%%%%%%%%%%%%

This article is the continuation of the previous work \cite{LewSab-13a} where we considered the nonlinear Hartree equation for infinitely many particles (but the main result of the present article does not rely on~\cite{LewSab-13a}). 

The Hartree equation can be written using the formalism of density matrices as
\begin{equation}\label{eq:Hartree}
 \left\{\begin{array}{rcl}
  i\partial_t\gamma & = & \big[-\Delta+w*\rho_\gamma,\gamma\big], \\
  \gamma(0) & = & \gamma_0.
 \end{array}\right.
\end{equation}
Here $\gamma(t)$ is the one-particle density matrix of the system, which is a bounded non-negative self-adjoint operator on $L^2(\R^d)$ with $d\geq1$, and $\rho_\gamma(t,x)=\gamma(t,x,x)$ is the density of particles in the system at time $t$. On the other hand $w$ is the interaction potential between the particles, which we assume to be smooth and fastly decaying at infinity.

The starting point of \cite{LewSab-13a} was the observation that \eqref{eq:Hartree} has many stationary states. Indeed, if $f\in L^\ii(\R_+,\R)$ is such that
$$\int_{\R^d}|f(|k|^2)|\,dk<+\ii,$$
then the operator 
$$\gamma_f:=f(-\Delta)$$
(the Fourier multiplier by $k\mapsto f(|k|^2)$) is a bounded self-adjoint operator which commutes with $-\Delta$ and whose density
$$\rho_{\gamma_f}(x)=(2\pi)^{-d} \int_{\R^d}f(|k|^2)\,dk,\qquad \forall x\in\R^d,$$ 
is constant. Hence, for $w\in L^1(\R^d)$, $w\ast\rho_{\gamma_f}$ is also constant, and $[w\ast \rho_{\gamma_f},\gamma_f]=0$. Therefore $\gamma(t)\equiv\gamma_f$ is a stationary solution to \eqref{eq:Hartree}. The purpose of~\cite{LewSab-13a} and of this article is to investigate the stability of these stationary states, under ``local perturbations''. We do not necessarily think of small perturbations in norm, but we typically think of $\gamma(0)-\gamma_f$ being compact.

The simplest choice is $f\equiv0$ which corresponds to the vacuum case. We are interested here in the case of $f\neq0$, describing an infinite, homogeneous gas containing infinitely many particles and with positive constant density $\rho_{\gamma_f}>0$.
Four important physical examples are the

\smallskip

\noindent$\bullet$ \emph{Fermi gas at zero temperature}: 
\begin{equation}
\gamma_f=\1(-\Delta\leq \mu),\qquad\mu>0;
\label{eq:Fermi-gas-zero-temp} 
\end{equation}

\medskip

\noindent$\bullet$ \emph{Fermi gas at positive temperature $T>0$}: 
\begin{equation}
\gamma_f=\frac{1}{e^{(-\Delta-\mu)/T}+1},\qquad \mu\in\R;
\label{eq:Fermi-gas-positive-temp} 
\end{equation}

\medskip

\noindent$\bullet$ \emph{Bose gas at positive temperature $T>0$}: 
\begin{equation}
\gamma_f=\frac{1}{e^{(-\Delta-\mu)/T}-1},\qquad \mu<0;
\label{eq:Bose-gas-positive-temp} 
\end{equation}

\medskip

\noindent$\bullet$ \emph{Boltzmann gas at positive temperature $T>0$}: 
\begin{equation}
\gamma_f=e^{(\Delta+\mu)/T},\qquad \mu\in\R.
\label{eq:Boltzmann-gas-positive-temp} 
\end{equation}
In the density matrix formalism, the number of particles in the system is given by $\tr\,\gamma$. It is clear that $\tr\,\gamma_f=+\ii$ in the previous examples since $\gamma_f$ is a translation-invariant (hence non-compact) operator. Because they contain infinitely many particles, these systems also have an infinite energy. In \cite{LewSab-13a}, we proved the existence of global solutions to the equation~\eqref{eq:Hartree} in the defocusing case $\widehat{w}\geq0$, when the initial datum $\gamma_0$ has a finite \emph{relative energy} counted with respect to the stationary states $\gamma_f$ given in~\eqref{eq:Fermi-gas-zero-temp}--\eqref{eq:Boltzmann-gas-positive-temp}, in dimensions $d=1,2,3$. We also proved the orbital stability of $\gamma_f$.

In this work, we are interested in the \emph{asymptotic stability} of $\gamma_f$. As usual for Schr\"odinger equations, we cannot expect strong convergence in norm and we will rather prove that $\gamma(t)\wto\gamma_f$ weakly as $t\to\pm\ii$, if the initial datum $\gamma_0$ is small enough. Physically, this means that a small defect added to the translation-invariant state $\gamma_f$ disappears for large times due to dispersive effects, and the system locally relaxes towards the homogeneous gas. More precisely, we are able to describe the exact behavior of $\gamma(t)$ for large times, by proving that 
$$e^{-it\Delta}\big(\gamma(t)-\gamma_f\big)e^{it\Delta}\underset{t\to\pm \ii}{\longrightarrow} Q_{\pm}$$
strongly in a Schatten space (hence for instance for the operator norm). This nonlinear scattering result means that the perturbation $\gamma(t)-\gamma_f$ of the homogeneous gas evolves for large times as in the case of free particles:
$$\gamma(t)-\gamma_f\underset{t\to\pm \ii}{\simeq} e^{it\Delta}Q_{\pm}e^{-it\Delta}\underset{t\to\pm \ii}{\wto}0.$$

If $f\equiv0$ and $\gamma_0=|u_0\rangle\langle u_0|$ is a rank-one orthogonal projection, then~\eqref{eq:Hartree} reduces to the well-known Hartree equation for one function
\begin{equation}\label{eq:Hartree-u}
\begin{cases}
i\partial_t u=(-\Delta+w*|u|^2)u,\\
u(0)=u_0.
\end{cases}
\end{equation}
There is a large literature about scattering for the nonlinear equation~\eqref{eq:Hartree-u}, see for instance  \cite{GinVel-80,Strauss-81,HayTsu-87,Mochizuki-89,GinVel-00c,Nakanishi-99}. The intuitive picture is that the nonlinear term is negligible for small $u$, since $w*|u|^2u$ is formally of order $3$. It is important to realize that this intuition does not apply in the case $f\neq0$ considered in this paper. Indeed the nonlinear term is not small and it behaves linearly with respect to the small parameter $\gamma-\gamma_f$:
\begin{equation}
\big[w*\rho_\gamma,\gamma\big]=\big[w*\rho_{\gamma-\gamma_f},\gamma\big]\simeq \big[w*\rho_{\gamma-\gamma_f},\gamma_f\big]\neq0. 
\label{eq:linear_intro}
\end{equation}
One of the main purpose of this paper is to rigorously study the linear response of the homogeneous Hartree gas $\gamma_f$ (the last term in~\eqref{eq:linear_intro}), which is a very important object in the physical literature, called the Lindhard function (see~\cite{Lindhard-54} and~\cite[Chap. 4]{GiuVig-05}). For a general $f$, our main result requires that the interaction potential $w$ is small enough, in order to control the linear term. Under the natural assumption that $f$ is strictly decreasing (as it is in the three physical examples~\eqref{eq:Fermi-gas-positive-temp}--\eqref{eq:Boltzmann-gas-positive-temp}), the condition can be weakened in the defocusing case $\widehat{w}\geq0$.

The paper is organized as follows. In the next section we state our main result and make several comments. In Section~\ref{sec:linear-response} we study the linear response in detail, before turning to the higher order terms in the expansion of the wave operator in Section~\ref{sec:higher-order}. Apart from the linear response, our method requires to treat separately the next $d-1$ terms of this expansion, in spacial dimension $d$. Even if all the other estimates are valid in any dimension, in this paper we only deal with the second order in dimension $d=2$.

%%%%%%%%%%%%%%%%%%%%%%%%%%%%%%%%%%%%%%%%%%
%%%%%%%%%%%%%%%%%%%%%%%%%%%%%%%%%%%%%%%%%%
\section{Main result}
%%%%%%%%%%%%%%%%%%%%%%%%%%%%%%%%%%%%%%%%%%
%%%%%%%%%%%%%%%%%%%%%%%%%%%%%%%%%%%%%%%%%%

In the whole paper, we denote by $\cB(\gH)$ the space of bounded operators on the Hilbert space $\gH$. The corresponding operator norm is $\|A\|$. We use the notation $\gS^p(\gH)$ for the Schatten space of all the compact operators $A$ on $\gH$ such that $\tr|A|^p<\ii$, with $|A|=\sqrt{A^*A}$, and use the norm $\norm{A}_{\gS^p(\gH)}:=(\tr|A|^p)^{1/p}$. We refer to~\cite{Simon-77} for the properties of Schatten spaces. The spaces $\gS^2(\gH)$ and $\gS^1(\gH)$ correspond to Hilbert-Schmidt and trace-class operators. We often use the shorthand notation $\cB$ and $\gS^p$ when the Hilbert space $\gH$ is clear from the context.

Our main result is the following.

\begin{theorem}[Dispersion and scattering in 2D]\label{thm:main}
Let $f:\R_+\to\R$ be such that 
\begin{equation}
\int_{0}^\ii (1+r^{\frac{k}2})|f^{(k)}(r)|\,dr<\ii\quad\text{for $k=0,...,4$}
\label{eq:derivees_f} 
\end{equation}
and $\gamma_f:=f(-\Delta)$. Denote by $\check{g}$ the Fourier inverse on $\R^2$ of $g(k)=f(|k|^2)$.
Let $w\in W^{1,1}(\R^2)$ be such that 
\begin{equation}
\norm{\check{g}}_{L^1(\R^2)}\norm{\hat w}_{L^\ii(\R^2)}<4\pi
\label{eq:condition_linear_response}
\end{equation}
or, if $f'<0$ a.e.~on $\R_+$, such that  
\begin{equation}
\max\left(\epsilon_g \widehat{w}(0)_+\;,\; \norm{\check{g}}_{L^1(\R^2)}\norm{(\hat w)_-}_{L^\ii(\R^2)}\right)<4\pi
\label{eq:condition_negative_part}
\end{equation}
where $(\hat w)_-$ is the negative part of $\hat w$ and $0\leq \epsilon_g\leq \norm{\check{g}}_{L^1(\R^2)}$ is a constant depending only on $g$ (defined later in Section~\ref{sec:linear-response}).

Then, there exists a constant $\epsilon_0>0$ (depending only on $w$ and $f$) such that, for any $\gamma_0\in\gamma_f+\gS^{4/3}$ with
 $$\norm{\gamma_0-\gamma_f}_{\gS^{4/3}}\le\epsilon_0,$$ there exists a unique $\gamma\in \gamma_f + C^0_t(\R,\gS^{2})$ solution to the Hartree equation \eqref{eq:Hartree} with initial datum $\gamma_0$, such that 
 $$\rho_\gamma-\rho_{\gamma_f}\in L^2_{t,x}(\R\times\R^2).$$
 Furthermore, $\gamma(t)$ scatters around $\gamma_f$ at $t=\pm\ii$, in the sense that there exists $Q_\pm\in\gS^4$ such that
\begin{multline}
\lim_{t\to\pm\ii}\norm{e^{-it\Delta}(\gamma(t)-\gamma_f)e^{it\Delta}-Q_\pm}_{\gS^4}\\
=\lim_{t\to\pm\ii}\norm{\gamma(t)-\gamma_f-e^{it\Delta}Q_\pm e^{-it\Delta}}_{\gS^4}=0. 
\label{eq:scattering}
\end{multline}
\end{theorem}

Before explaining our strategy to prove Theorem~\ref{thm:main}, we make some comments.

First we notice that the gases at positive temperature~\eqref{eq:Fermi-gas-positive-temp},~\eqref{eq:Bose-gas-positive-temp} and~\eqref{eq:Boltzmann-gas-positive-temp} are all covered by the theorem with the condition~\eqref{eq:condition_negative_part}, since the corresponding $f$ is smooth, strictly decreasing and exponentially decaying at infinity. Our result does not cover the Fermi gas at zero temperature~\eqref{eq:Fermi-gas-zero-temp}, however. We show in Section~\ref{sec:linear-response} that its linear response is unbounded and it is a challenging task to better understand its dynamical stability.

The next remark concerns the assumption~\eqref{eq:condition_linear_response} which says that the interactions must be small or, equivalently, that the gas must contain few particles having a small momentum (if $\check{g}\geq0$, then the condition can be written $f(0)\norm{\widehat w}_{L^\ii(\R^2)}<2$ and hence $f(|k|^2)$ must be small for small $k$). Our method does not work without the condition~\eqref{eq:condition_linear_response} if no other information on $w$ and $f$ is provided. However, under the natural additional assumption that $f$ is strictly decreasing, we can replace the condition~\eqref{eq:condition_linear_response} by the weaker condition~\eqref{eq:condition_negative_part}. The latter says that the negative part of $\widehat{w}$ and the value at zero of the positive part should be small (with a better constant for the latter). We will explain later where the condition~\eqref{eq:condition_negative_part} comes from, but we mention already that we are not able to deal with an arbitrary large potential $\widehat{w}$ in a neighborhood of the origin, even in the defocusing case. We also recall that the focusing or defocusing character of our equation is governed by the sign of $\widehat{w}$ and not of $w$, as it is for~\eqref{eq:Hartree-u}. This is seen from the sign of the nonlinear term 
$$\int_{\R^d}\int_{\R^d}w(x-y)\rho_{\gamma-\gamma_f}(x)\rho_{\gamma-\gamma_f}(y)\,dx\,dy=(2\pi)^{d/2}\int_{\R^d}\hat{w}(k)|\hat{\rho_{\gamma-\gamma_f}}(k)|^2\,dk$$
which appears in the relative energy of the system (see~\cite[Eq.~(9)--(10)]{LewSab-13a}).

Let us mention that our results hold for \emph{small} initial data, where the smallness is not only qualitative (meaning that $\gamma_0-\gamma_f\in \gS^{4/3}$ for instance), but also quantitative since we need that $\norm{\gamma_0-\gamma_f}_{\gS^{4/3}}$ be small enough. This is a well-known restriction, coming from our method of proof, based on a fixed point argument. The literature on nonlinear Schrödinger equations suggests that, in order to remove this smallness assumption, one would need some assumption on $w$ like $\widehat{w}\geq0$, as well as some additional (almost) conservation laws \cite{Cazenave}. Our study of the linear response operator however indicates that the situation is involved and more information on the momentum distribution $f$ is certainly also necessary.

We finally note that in our previous article~\cite{LewSab-13a}, we proved the existence of global solutions under the assumption that the initial state $\gamma_0$ has a finite relative entropy with respect to $\gamma_f$ (and for $f$ being one of the physical examples \eqref{eq:Fermi-gas-zero-temp}--\eqref{eq:Boltzmann-gas-positive-temp}). By the Lieb-Thirring inequality (see~\cite{FraLewLieSei-11,FraLewLieSei-12} and~\cite{LewSab-13a}), this implies that $\rho_{\gamma(t)}-\rho_{\gamma_f}\in L^\ii_t(L^2_x)$. By interpolation we therefore get that $\rho_{\gamma(t)}-\rho_{\gamma_f}\in L^p_t(L^2_x)$ for every $2\leq p\leq\ii$. This requires of course that the initial perturbation $\gamma_0-\gamma_f$ be small in $\gS^{4/3}$. Our method does not allow to replace this condition by the fact that $\gamma_0$ has a small relative entropy with respect to $\gamma_f$.

\medskip

We now explain our strategy for proving Theorem~\ref{thm:main}. The idea of the proof relies on a fixed point argument, in the spirit of \cite[Sec. 5]{LewSab-13a}. If we can prove that $\rho_\gamma-\rho_{\gamma_f}\in L^{2}_{t,x}(\R_+\times\R^2)$, then we deduce from \cite{Yajima-87,FraLewLieSei-13} that there exists a family of unitary operators $U_V(t)\in C^0_t(\R_+,\cB)$ on $L^2(\R^2)$ such that 
$$\gamma(t)=U_V(t)\gamma_0U_V(t)^*,$$
for all $t\in\R_+$. We furthermore have
$$U_V(t)=e^{it\Delta}\cW_V(t),$$
where $\cW_V(t)$ is the wave operator. By iterating Duhamel's formula, the latter can be expanded in a series as
\begin{equation}
\cW_V(t)=1+\sum_{n\ge1}\cW_V^{(n)}(t)
\label{eq:expansion_wave_op} 
\end{equation}
with
\begin{multline*}
    \cW_V^{(n)}(t):=(-i)^n\int_0^t\,dt_n\int_0^{t_n}\,dt_{n-1}\cdots\int_0^{t_2}\,dt_1\times\\
    \times e^{-it_n\Delta}V(t_n)e^{i(t_n-t_{n-1})\Delta}\cdots e^{i(t_2-t_1)\Delta}V(t_1)e^{it_1\Delta}.
 \end{multline*}
 The idea is to find a solution to the nonlinear equation
 \begin{equation}\label{eq:rhoQ1}
  \rho_Q(t)=\rho\left[e^{it\Delta}\cW_{w*\rho_Q}(t)(\gamma_f+Q_0)\cW_{w*\rho_Q}(t)^*e^{-it\Delta}\right]-\rho_{\gamma_f},
 \end{equation}
 by a fixed point argument on the variable $\rho_Q\in L^2_{t,x}(\R\times\R^2)$, where $Q:=\gamma-\gamma_f$ and $Q_0=\gamma_0-\gamma_f$.

Inserting the expansion~\eqref{eq:expansion_wave_op} of the wave operator $\cW_V$, the nonlinear equation~\eqref{eq:rhoQ} may be written as 
 \begin{equation}\label{eq:rhoQ}
  \rho_Q=\rho\left[e^{it\Delta}Q_0e^{-it\Delta}\right]-\cL(\rho_Q)+\cR(\rho_Q),
 \end{equation}
where $\cL$ is linear and $\cR(\rho_Q)$ contains higher order terms. The sign convention for $\cL$ is motivated by the stationary case~\cite{FraLewLieSei-12}.
The linear operator $\cL$ can be written 
$$\cL=\cL_1+\cL_2$$
where
$$\cL_1(\rho_Q)=-\rho\left[e^{it\Delta}(\cW^{(1)}_{w*\rho_Q}(t)\gamma_f+\gamma_f\cW^{(1)}_{w*\rho_Q}(t)^*)e^{-it\Delta}\right]$$
and
$$\cL_2(\rho_Q)=-\rho\left[e^{it\Delta}(\cW^{(1)}_{w*\rho_Q}(t)Q_0+Q_0\cW^{(1)}_{w*\rho_Q}(t)^*)e^{-it\Delta}\right].$$
Note that $\cL_2$ depends on $Q_0$ and it can always be controlled by adding suitable assumptions on $Q_0$. On the other hand, the other linear operator $\cL_1$ does not depend on the studied solution, it only depends on the functions $w$ and $f$. 

In Section~\ref{sec:linear-response}, we study the linear operator $\cL_1$ in detail, and we prove that it is a space-time Fourier multiplier of the form $\widehat{w}(k)m_f(\omega,k)$ where $m_f$ is a famous function in the physics literature called the \emph{Lindhard function}~\cite{Lindhard-54,Mihaila-11,GiuVig-05}), which only depends on $f$ and $d$. We particularly investigate when $\cL_1$ is bounded on $L^p_{t,x}(\R\times\R^2)$ and we show it is the case when $w$ and $f$ are sufficiently smooth. For the Fermi sea~\eqref{eq:Fermi-gas-zero-temp}, we prove that $\cL_1$ is unbounded on $L^2_{t,x}$.

The next step is to invert the linear part by rewriting the equation~\eqref{eq:rhoQ} in the form
 \begin{equation}\label{eq:rhoQ2}
\rho_Q=(1+\cL)^{-1}\Big(\rho\left[e^{it\Delta}Q_0e^{-it\Delta}\right]+\cR(\rho_Q)\Big)
 \end{equation}
and to apply a fixed point method. In the time-independent case, a similar technique was used for the Dirac sea in~\cite{HaiLewSer-05a}. In order to be able to invert the Fourier multiplier $\cL_1$, we need that
\begin{equation}
\boxed{\phantom{\int}\min_{(\omega,k)\in\R\times\R^2}\left|\widehat{w}(k)m_f(\omega,k)+1\right|>0.\phantom{\int}}
\label{eq:condition_w_precise} 
\end{equation}
Then $1+\cL=1+\cL_1+\cL_2$ is invertible if $Q_0$ is small enough. In Section~\ref{sec:linear-response} we prove the simple estimate
$$|m_f(\omega,k)|\leq (4\pi)^{-1}\norm{\check{g}}_{L^1(\R^2)}$$
and this leads to our condition~\eqref{eq:condition_linear_response}. If $f$ is strictly decreasing, then we are able to prove that the imaginary part of $m_f(k,\omega)$ is never 0 for $k\neq0$ or $\omega\neq0$. Since $m_f(\omega,k)$ has a fixed sign for $\omega=0$ and $k=0$, everything boils down to investigating the properties of $m_f$ at $(\omega,k)=(0,0)$. At this point $m_f$ will usually not be continuous, and it can take both positive and negative values. We have
$$\limsup_{(\omega,k)\to(0,0)}\Re \,m_f(\omega,k)=(4\pi)^{-1}\norm{\check{g}}_{L^1(\R^2)}$$
and we denote
$$\liminf_{(\omega,k)\to(0,0)}\Re \,m_f(\omega,k):=-(4\pi)^{-1}\epsilon_g,$$
leading to our condition~\eqref{eq:condition_negative_part}. It is well-known in the physics literature that the imaginary part of the Lindhard function plays a crucial role in the dynamics of the homogeneous Fermi gas. In our rigorous analysis it is used to invert the linear response operator outside of the origin. The behavior of $m_f(\omega,k)$ for $(\omega,k)\to(0,0)$ is however involved and $1+\cL_1$ is not invertible if $\widehat{w}(0)>4\pi/\epsilon_g$ or $\widehat{w}(0)<-4\pi/\norm{\check{g}}_{L^1(\R^2)}$.

For the Fermi gas at zero temperature~\eqref{eq:Fermi-gas-zero-temp} we will prove that the minimum in~\eqref{eq:condition_w_precise} is always zero, except when $\widehat{w}$ vanishes sufficiently fast at the origin, this means that $1+\cL_1$ is never invertible. It is an interesting open question to understand the asymptotic stability of the Fermi sea.

Once the linear response $\cL$ has been inverted, it remains to study the zeroth order term $\rho\left[e^{it\Delta}Q_0e^{-it\Delta}\right]$ and the higher order terms contained in $\cR(\rho_Q)$. At this step we use a recent Strichartz estimate in Schatten spaces which is due to Frank, Lieb, Seiringer and the first author. 
 
  \begin{theorem}[Strichartz estimate on wave operator~{\cite[Thm 3]{FraLewLieSei-13}}]\label{thm:est-Wn}
  Let $d\ge1$, $1+d/2\le q<\ii$, and $p$ such that $2/p+d/q=2$. Let also $0<\epsilon<1/p$. Then, there exists $C=C(d,p,\epsilon)>0$ such that for any $V\in L^p_t(\R,L^q_x(\R^d))$ and any $t\in\R$, we have the estimates
  \begin{equation}\label{eq:est-W1}
\left\|\cW^{(1)}_V(t)\right\|_{\gS^{2q}}\le C\|V\|_{L^p_t L^q_x}
  \end{equation}
and
  \begin{equation}\label{eq:est-Wn}
\forall n\geq2,\qquad    \left\|\cW^{(n)}_V(t)\right\|_{\gS^{2\left\lceil\frac{q}{n}\right\rceil}}\le\frac{C^n}{(n!)^{\frac{1}{p}-\epsilon}}\|V\|_{L^p_t L^q_x}^n.
  \end{equation}
 \end{theorem}

The estimate~\eqref{eq:est-W1} is the dual version of 
\begin{equation}\label{eq:Strichartz}
 \norm{\rho_{e^{it\Delta}Ae^{-it\Delta}}}_{L^p(\R,L^q(\R^d))}\le C\|A\|_{\gS^{\frac{2q}{q+1}}},
\end{equation}
for any $(p,q)$ such that $2/p+d/q=d$ and $1\le q\le 1+2/d$, see~\cite[Thm. 1]{FraLewLieSei-13}. The estimate~\eqref{eq:Strichartz} is useful to deal with the first order term involving $Q_0$ in~\eqref{eq:rhoQ2}, leading to the natural condition that $Q_0\in \gS^{4/3}$ in dimension $d=2$ with $p=q=2$.
 
In dimension $d$, it seems natural to prove that $\rho_Q\in L^{1+2/d}_{t,x}(\R\times\R^d)$. The estimate~\eqref{eq:est-Wn} turns out to be enough to deal with the terms of order $\geq d+1$ but it does not seem to help for the terms of order $\leq d$, because the wave operators $\cW^{(n)}_V$ with small $n$ belong to a Schatten space with a too large exponent. Apart from the linear response, we are therefore left with $d-1$ terms for which a more detailed computation is necessary. We are not able to do this in any dimension (the number of such terms grows with $d$), but we can deal with the second order term in dimension $d=2$,
 $$\rho\left[e^{it\Delta}(\cW^{(2)}_{w*\rho_Q}(t)\gamma_f+\cW^{(1)}_{w*\rho_Q}(t)\gamma_f\cW^{(1)}_{w*\rho_Q}(t)^*+\gamma_f\cW^{(2)}_{w*\rho_Q}(t)^*)e^{-it\Delta}\right],$$
which then finishes the proof of the theorem in this case. The second-order term is the topic of Section \ref{sec:second-order}. 

Even if our final result only covers the case $d=2$, we have several estimates in any dimension $d\geq2$. With the results of this paper, only the terms of order $2$ to $d$ remain to be studied to obtain a result similar to Theorem~\ref{thm:main} (with $\rho_\gamma-\rho_{\gamma_f}\in L^{1+2/d}_{t,x}(\R\times\R^d)$) in dimensions $d\geq3$.

%%%%%%%%%%%%%%%%%%%%%%%%%%%%%%%%%%%%%%%%%%
%%%%%%%%%%%%%%%%%%%%%%%%%%%%%%%%%%%%%%%%%%
\section{Linear response theory}\label{sec:linear-response}
%%%%%%%%%%%%%%%%%%%%%%%%%%%%%%%%%%%%%%%%%%
%%%%%%%%%%%%%%%%%%%%%%%%%%%%%%%%%%%%%%%%%%

%%%%%%%%%%%%%%%%%%%%%%%%%%%%%%%%%%%%%%%%%%
\subsection{Computation of the linear response operator}
%%%%%%%%%%%%%%%%%%%%%%%%%%%%%%%%%%%%%%%%%%

As we have explained before, we deal here with the linear response $\cL_1$ associated with the homogeneous state $\gamma_f$. The first order in Duhamel's formula is defined by
$$Q_1(t):=-i\int_0^t e^{i(t-t')\Delta}[w*\rho_{Q(t')},\gamma_f]e^{i(t'-t)\Delta}\,dt'.$$
We see that it is a linear expression in $\rho_Q$, and we compute its density as a function of $\rho_Q$.

\begin{proposition}[Uniform bound on $\cL_1$]\label{prop:linear-response}
Let $d\ge1$, $f\in L^\ii(\R_+,\R)$ such that $\int_{\R^d}|f(k^2)|\,dk<+\ii$, and $w\in L^1(\R^d)$. Then, the linear operator $\cL_1$ defined for all $\phi\in\cD(\R_+\times\R^d)$ by
$$\cL_1(\phi)(t):=-\rho\big[Q_1(t)\big]=\rho\left[i\int_0^t e^{i(t-t')\Delta}[w*\phi(t'),\gamma_f]e^{i(t'-t)\Delta}\,dt'\right]$$
is a space-time Fourier multiplier by the kernel $K^{(1)}=\widehat{w}(k)\,m_f(\omega,k)$, where 
\begin{equation}
\boxed{
\left[\cF^{-1}_\omega m_f\right](t,k):=2\1_{t\ge0}\sqrt{2\pi}\sin(t|k|^2)\check{g}(2tk)
}
\end{equation}
(we recall that $g(k):=f(k^2)$ and that $\check{g}$ is its Fourier inverse). This means that for all $\phi\in\cD(\R_+\times\R^d)$, we have
$$\cF_{t,x}\left[\cL_1(\phi)\right](\omega,k)=\hat{w}(k)m_f(\omega,k)\left[\cF_{t,x}\phi\right](\omega,k),\quad\forall(\omega,k)\in\R\times\R^d$$
where $\cF_{t,x}$ is the space-time Fourier transform.
Furthermore, if $\int_0^\ii|x|^{2-d}|\check{g}(x)|\,dx<\ii$, then $m_f\in L^\ii_{\omega,k}(\R\times\R^d)$ and we have the explicit estimates
\begin{equation}
 \|m_f\|_{L^\ii_{\omega,k}}\le \frac{1}{2|\mathbb{S}^{d-1}|}\left(\int_{\R^d}\frac{|\check{g}(x)|}{|x|^{d-2}}\,dx\right)
\label{eq:estim_m_f}
\end{equation}
and
\begin{equation}
 \|\cL_1\|_{L^2_{t,x}\to L^2_{t,x}}\le \frac{\norm{\hat{w}}_{L^\ii}}{2|\mathbb{S}^{d-1}|}\left(\int_{\R^d}\frac{|\check{g}(x)|}{|x|^{d-2}}\,dx\right).
\label{eq:estim_L_1}
\end{equation}
\end{proposition}

\begin{proof}
Let $\phi\in\cD(\R_+\times\R^d)$. In order to compute $\cL(\phi)$, we use the relation 
$$\int_0^\ii\tr[W(t,x)Q_1(t)]\,dt=\int_0^\ii\int_{\R^d}W(t,x)\rho_{Q_1}(t,x)\,dx\,dt,$$
valid for any function $W\in \cD(\R_+\times\R^d)$.
This leads to
\begin{multline*}
 \int_0^\ii\int_{\R^d}W(t,x)\rho_{Q_1}(t,x)\,dx\,dt=\frac{-i}{(2\pi)^{d}}\int_0^\ii\int_0^t\int_{\R^d}\int_{\R^d} e^{-2i(t-t')k\cdot\ell}\times\\
\times\hat{W}(t,-k)\hat{V}(t',k)(g(\ell-k/2)-g(\ell+k/2)) d\ell\,dk\,dt'\,dt,
\end{multline*}
where $g(k):=f(k^2)$ and $V=w*\phi$. Computing the $\ell$-integral gives
\begin{multline*}
\int_{\R^d}e^{-2i(t-t')k\cdot\ell}(g(\ell-k/2)-g(\ell+k/2))\,d\ell\\=-(2\pi)^{d/2}2i\sin((t-t')|k|^2)\check{g}(2(t-t')k). 
\end{multline*}
Hence, using that $\hat{V}=(2\pi)^{d/2}\hat{w}\hat{\phi}$, we find that
\begin{multline*}
   \int_0^\ii\int_{\R^d}W(t,x)\rho_{Q_1}(t,x)\,dx\,dt\\
   =-2\int_0^\ii\int_0^t\int_{\R^d}\sin((t-t')|k|^2)\check{g}(2(t-t')k)\widehat{w}(k)\hat{W}(t,-k)\hat{\phi}(t',k)\,dk\,dt'\,dt.
\end{multline*}
Since $g$ is radial, then $\check{g}$ is also radial and we have
\begin{eqnarray*}
|m_f(\omega,k)| & \leq & 2\int_0^\ii|\sin(t|k|^2)||\check{g}(2t|k|)|\,dt\\
 & \leq & 2\int_0^\ii\frac{|\sin(t|k|)|}{|k|}|\check{g}(2t)|\,dt\\
 & \le & \frac{1}{2}\int_0^\ii r|\check{g}(r)|\,dr.
\end{eqnarray*}
This ends the proof of the proposition.
\end{proof}

We now make several remarks about the previous result.

First, the physical examples for $g$ are
$$g(k)=\begin{cases}
   \displaystyle \1(|k|^2\le\mu), & \mu>0,\\
   \displaystyle e^{-(|k|^2-\mu)/T}, &T>0,\ \mu\in\R,\\
   \displaystyle \frac{1}{e^{(|k|^2-\mu)/T}+1}, & T>0,\ \mu\in\R,\\
   \displaystyle \frac{1}{e^{(|k|^2-\mu)/T}-1}, & T>0,\ \mu<0.
  \end{cases}
$$
In the last three choices, $g$ is a Schwartz function hence $\check{g}\in L^1(\R^d)$. For the first choice of $g$ (Fermi sea at zero temperature), we have $\check{g}(r)\sim r^{-1}\sin r$, which obviously does not verify $r\check{g}(r)\in L^1(0,+\ii)$. 

Then, we remark that~\eqref{eq:estim_m_f} is optimal without more assumptions on $f$. Indeed, for $\omega=0$ and small $k$ we find
\begin{align*}
 m_f(0,k)&=2\int_0^\ii\sin(t|k|^2)\check{g}(2tk)\,dt\\
&\underset{k\to0}\longrightarrow \frac{1}2\int_0^\ii r\check{g}(r)\,dr=\frac{1}{2|\mathbb{S}^{d-1}|}\int_{\R^d}\frac{\check{g}(x)}{|x|^{d-2}}\,dx.
\end{align*}
We conclude that~\eqref{eq:estim_m_f} is optimal if $\check{g}$ has a constant sign (for instance $f$ is decreasing, as in the physical examples~\eqref{eq:Fermi-gas-positive-temp}--\eqref{eq:Boltzmann-gas-positive-temp}). Similarly,~\eqref{eq:estim_L_1} is optimal if both $\check{g}$ and $w$ have a constant sign (then $|\widehat{w}(0)|=\norm{\widehat{w}}_{L^\ii}$). 

In general, the function $m_f$ is complex-valued and it is not an easy task to determine when $\hat{w}(k)m_f(\omega,k)$ stays far from $-1$. Since the stationary linear response is real, $\Im m_f(0,k)\equiv0$, the condition should at least involve the maximum or the minimum of $m_f$ on the set $\{\omega=0\}$, depending on the sign of $\widehat{w}$. Even if the function $m_f$ is bounded on $\R\times\R^d$ by~\eqref{eq:estim_m_f}, it will usually not be continuous at the point $(0,0)$. Under the additional condition that $f$ is strictly decreasing, we are able to prove that 
$$\{\Im m_f(\omega,k)=0\}=\{\omega=0\}\cup\{k=0\}$$
and this can be used to replace the assumption on $\widehat{w}$ by one on $(\widehat{w})_-$ and $\widehat{w}(0)_+$.
In order to explain this, we first compute $m_f$ in the case of a Fermi gas at zero temperature, $f(k^2)=\1(|k|^2\le\mu)$. 

\begin{proposition}[Linear response at zero temperature]
Let $d\ge1$ and $\mu>0$. Then, for the Fermi sea at zero temperature $\gamma_f=\1(-\Delta\le\mu)$, the corresponding Fourier multiplier $m_f(\omega,k):=m^{\rm F}_{d}(\mu,\omega,k)$ of the linear response operator in dimension $d$ is given by
\begin{multline}
m_1^{\rm F}(\mu,\omega,k)=\frac{1}{2\sqrt{2\pi}|k|}\log\left|\frac{(|k|^2+2|k|\sqrt{\mu})^2-\omega^2}{(|k|^2-2|k|\sqrt{\mu})^2-\omega^2}\right|\\ + i\frac{\sqrt{\pi}}{2\sqrt{2}|k|}\bigg\{\1\left(|\omega+|k|^2|\leq 2\sqrt\mu|k|\right)-\1\left(|\omega-|k|^2|\leq 2\sqrt\mu|k|\right)\bigg\}
\label{eq:m_f_1D}
\end{multline}
for $d=1$, by
\begin{multline}
m_2^{\rm F}(\mu,\omega,k)=\frac 14\bigg\{2-\frac{{\rm sgn}(|k|^2+\omega)}{|k|^2}\Big((|k|^2+\omega)^2-4\mu |k|^2\Big)^{\frac12}_+\\ 
- \frac{{\rm sgn}(|k|^2-\omega)}{|k|^2}\Big((|k|^2-\omega)^2-4\mu |k|^2\Big)^{\frac12}_+\bigg\}\\
+i\frac{1}{2|k|^2}\bigg\{\Big((|k|^2+\omega)^2-4\mu |k|^2\Big)^{\frac12}_- - \Big((|k|^2-\omega)^2-4\mu |k|^2\Big)^{\frac12}_-\bigg\}.
\label{eq:m_f_2D}
\end{multline}
for $d=2$, and by
\begin{align}
 m_d^{\rm F}(\mu,\omega,k)&=\frac{|\mathbb{S}^{d-2}|\mu^{\frac{d-1}{2}}}{(2\pi)^{\frac{d-1}2}}\int_{0}^1  m^{\rm F}_1\big(\mu(1-r^2),\omega,k\big) r^{d-2}\,dr,&\text{for $d\geq2$,}\nonumber\\
&=\frac{|\mathbb{S}^{d-3}|\mu^{\frac{d-2}{2}}}{(2\pi)^{\frac{d-2}2}}\int_{0}^1  m^{\rm F}_2\big(\mu(1-r^2),\omega,k\big) r^{d-3}\,dr,&\text{for $d\geq3$.}\label{eq:m_f_3D}
\end{align}
\end{proposition}

The formula for $m^{\rm F}_{d}$ is well known in the physics literature (see~\cite{Lindhard-54}, \cite{Mihaila-11} and~\cite[Chap. 4]{GiuVig-05}). It is also possible to derive an explicit expression for $m_{3}^{\rm F}(\mu,\omega,k)$, see~\cite[Chap. 4]{GiuVig-05}. We remark that $m_{d}^{\rm F}(\mu,0,k)$ coincides with the time-independent linear response computed in~\cite[Thm 2.5]{FraLewLieSei-12}. 

From the formulas we see that the real part of $m^{\rm F}_{d}$ can have both signs. It is always positive for $\omega=0$ and it can take negative values for $\omega\neq0$. For instance, in dimension $d=2$, on the curve $\omega=|k|^2+2\sqrt{\mu}|k|$ the imaginary part vanishes and we get
\begin{equation}
m_{2}^{\rm F}\big(\mu,|k|^2+2\sqrt{\mu}|k|,k\big)=\frac 12\left(1-\sqrt{1+\frac{2\sqrt\mu}{|k|}}\right)\underset{k\to0}{\longrightarrow}-\ii.
\label{eq:curve_omega_k}
\end{equation}
In particular, if $\hat{w}(k)/\sqrt{|k|}\to+\ii$ when $k\to0$, then $\widehat{w}(k)m_f(|k|^2+2\sqrt{\mu}|k|,k)\to-\ii$ when $|k|\to 0$. Since on the other hand $\widehat{w}(k)m_f(|k|^2+2\sqrt{\mu}|k|,k)\to0$ when $|k|\to\ii$, we conclude that the function must cross $-1$, and $(1+\cL_1)^{-1}$ is not bounded.

An important feature of $m_d^{\rm F}$ which we are going to use in the positive temperature case, is that the imaginary part $\Im m_{d}^{\rm F}(\mu,\omega,k)$ has a constant sign on $\{\omega>0\}$ and on  $\{\omega<0\}$. Before we discuss this in detail, we provide the proof of the proposition.

\begin{proof}
 First, a calculation shows that the Fourier inverse $ \check{g}_1$ of the radial $g$ in dimension $d=1$ is given by
\begin{equation}
 \check{g}_1(\mu,x)=\sqrt{\frac{2}{\pi}}\frac{\sin(\sqrt{\mu}|x|)}{|x|}.
\label{eq:g_Fourier_inverse_1D}
\end{equation}
In dimension $d\geq2$ we can write 
\begin{align}
\check{g}_d(\mu,|x|)=&\frac{1}{(2\pi)^{d/2}}\int_{\R^d}\1(|k|^2\leq \mu) e^{i k\cdot x}\nonumber\\
=&\frac{1}{(2\pi)^{d/2}}\int_{\R}dk_1\int_{\R^{d-1}}dk_\perp \1(|k_1|^2\leq \mu-|k_\perp|^2) e^{i k_1|x|}\nonumber\\
=&\frac{|\mathbb{S}^{d-2}|\mu^{\frac{d-1}{2}}}{(2\pi)^{d/2}}\int_{\R}dk_1\int_{0}^1  \1\big(|k_1|^2\leq \mu(1-r^2)\big) e^{i k_1|x|}r^{d-2}dr\nonumber\\
=&\frac{|\mathbb{S}^{d-2}|\mu^{\frac{d-1}{2}}}{(2\pi)^{\frac{d-1}2}}\int_{0}^1  \check{g}_1\big(\mu(1-r^2),|x|\big)r^{d-2}\,dr\nonumber\\
=&\frac{2|\mathbb{S}^{d-2}|}{(2\pi)^{d/2}}\frac{\mu^{\frac{d-1}{2}}}{|x|}\int_0^1\sin(\sqrt{\mu}|x|\sqrt{1-r^2})r^{d-2}\,dr.\label{eq:g_Fourier_inverse_2D}
\end{align}
Similarly, we have in dimension $d\geq3$
\begin{equation}
 \check{g}_d(\mu,|x|)=\frac{|\mathbb{S}^{d-3}|\mu^{\frac{d-2}{2}}}{(2\pi)^{\frac{d-2}2}}\int_{0}^1  \check{g}_2\big(\mu(1-r^2),|x|\big)r^{d-3}\,dr.
\label{eq:g_Fourier_inverse_3D}
\end{equation}
Now we can compute the multiplier $m^{\rm F}_d(\mu,\omega,k)$ for $d=1,2$. We start with $d=1$ for which we have
$$[\cF^{-1}_\omega m_{f,1}](t,k)=4\1_{t\ge0}\frac{\sin(t|k|^2)\sin(2\sqrt{\mu}t|k|)}{2t|k|}.$$
There remains to compute the time Fourier transform. We use the formula valid for any $a,b\in\R$,
\begin{multline*}
\int_0^\ii\frac{\sin(at)\sin(bt)}{t}e^{-it\omega}\d{t}=\frac{1}{4}\log\left|\frac{(a+b)^2-\omega^2}{(a-b)^2-\omega^2}\right|\\
+i\frac{\pi}{8}\left(\sgn(a-b-\omega)-\sgn(a+b-\omega)+\sgn(a+b+\omega)-\sgn(a-b+\omega)\right), 
\end{multline*}
and obtain~\eqref{eq:m_f_1D}. To provide the more explicit expression in dimension 2, we use this time the formula 
 $$\forall a\in\R,\qquad \frac{1}{a}\int_0^1\log\frac{|a+2\sqrt{1-r^2}|}{|a-2\sqrt{1-r^2}|}\d{r}=\frac{\pi}{2}-\frac{\pi}{2}\left(1-\frac{4}{a^2}\right)_+^{1/2},$$
which leads to the claimed form~\eqref{eq:m_f_2D} of $m_{2}^{\rm F}(\mu,\omega,k)$.
\end{proof}

Now we will use the imaginary part of $m_d^{\rm F}$ to show that $1+\cL_1$ is invertible with bounded inverse when $\widehat{w}\geq0$ with $\widehat{w}(0)$ not too large, and when $f$ is strictly decreasing.

\begin{corollary}[$1+\cL_1$ is always invertible in the defocusing case]\label{cor:linear-response_defocusing}
Let $d\ge1$ and $f\in L^\ii(\R_+,\R)$ such that $\int_0^\ii (r^{d/2-1}|f(r)|+|f'(r)|)\,dr<\ii$ and $f'(r)<0$ for all $r>0$. Assume furthermore that 
$\int_{\R^d}|x|^{2-d} |\check{g}(x)|\,dx<\ii$ with $g(k)=f(|k|^2)$. If $w\in L^1(\R^d)$ is an even function such that 
\begin{equation}
 \norm{(\hat{w})_-}_{L^\ii}\left(\int_{\R^d}\frac{|\check{g}(x)|}{|x|^{d-2}}\,dx\right)<2|\mathbb{S}^{d-1}|,
\label{eq:condition_negative_part_bis}
\end{equation}
and such that
\begin{equation}
\epsilon_g \hat{w}(0)_+<2|\mathbb{S}^{d-1}|,\quad\text{where}\quad \epsilon_g:=-\liminf_{(\omega,k)\to(0,0)}\frac{\Re m_f(\omega,k)}{2|\mathbb{S}^{d-1}|},
\label{eq:condition_negative_part_bis2}
\end{equation}
then we have
$$\min_{(\omega,k)\in\R\times\R^d}|\widehat{w}(k)m_f(\omega,k)+1|>0$$
and $(1+\cL_1)$ is invertible on $L^2_{t,x}(\R\times\R^d)$ with bounded inverse.
\end{corollary}

\begin{proof}
First we recall that $m_f$ is uniformly bounded by~\eqref{eq:estim_m_f}. Therefore we only have to look at the set 
$$A=\left\{k\in\R^d\ :\ |\widehat{w}(k)|\geq  \frac{1}{4|\mathbb{S}^{d-1}|}\left(\int_{\R^d}\frac{|\check{g}(x)|}{|x|^{d-2}}\,dx\right)\right\}.$$
On the complement of $A$, we have $|\widehat{w}\,m_f+1|\geq 1/2$. Since $\widehat{w}(k)\to0$ when $|k|\to\ii$, then $A$ is a compact set. Next, from the integral formula 
$$f(|k|^2)=-\int_0^\ii\1(|k|^2\leq s) f'(s)\,ds,$$ 
we infer that
$$m_f(\omega,k)=-\int_0^\ii m_d^{\rm F}(s,\omega,k) f'(s)\,ds.$$
This integral representation can be used to prove that $m_f$ is continuous on $\R\times\R_+\setminus\{(0,0)\}$. In general, the function $m_f$ is not continuous at $(0,0)$, however.

Since $m_d^{\rm F}(s,0,k)\geq0$ for all $k$ and $s\geq0$, we conclude that $m_f(0,k)\geq0$ and that 
$$m_f(0,k)\widehat{w}(k)\geq -m_f(0,k)\widehat{w}(k)_-\geq -\norm{\widehat{w}_-}_{L^\ii}\frac{1}{2|\mathbb{S}^{d-1}|}\left(\int_{\R^d}\frac{|\check{g}(x)|}{|x|^{d-2}}\,dx\right),$$
due to~\eqref{eq:estim_m_f}. In particular, 
$$|m_f(0,k)\widehat{w}(k)+1|\geq 1-\norm{\widehat{w}_-}_{L^\ii}\frac{1}{2|\mathbb{S}^{d-1}|}\left(\int_{\R^d}\frac{|\check{g}(x)|}{|x|^{d-2}}\,dx\right)>0$$
due to our assumption on $(\widehat{w})_-$. Similarly, we have $m_f(\omega,0)=0$ for all $\omega\neq0$ and therefore $m_f(\omega,0)\widehat{w}(0)+1=1$ is invertible on $\{k=0,\omega\neq0\}$. 

Now we look at $k\neq0$ and $\omega>0$ and we prove that the imaginary part of $m_f$ never vanishes. We write the argument for $d=1$, as it is very similar for $d\geq2$, using the integral representation~\eqref{eq:m_f_2D}. We  have
\begin{multline*}
\Im m_f(\omega,k)
=\frac{\sqrt{\pi}}{2\sqrt{2}|k|}\times\\
\times\int_0^\ii  
\bigg\{\1\left((\omega-|k|^2)^2\leq 4s|k|^2\right)-\1\left((\omega+|k|^2)^2\leq 4s|k|^2\right)\bigg\} f'(s)\,ds.
\end{multline*}
The difference of the two Heaviside functions is always $\geq0$ for $\omega>0$. Furthermore, it is equal to 1 for all $s$ in the interval
$$\frac{(\omega-|k|^2)^2}{4|k|^2}\leq s\leq \frac{(\omega+|k|^2)^2}{4|k|^2}.$$
Therefore we have
$$\Im m_f(\omega,k)\leq \frac{\sqrt{\pi}}{2\sqrt{2}|k|}\int_{\frac{(\omega-|k|^2)^2}{4|k|^2}}^{\frac{(\omega+|k|^2)^2}{4|k|^2}}  
 f'(s)\,ds<0$$
for all $\omega>0$ and $k\neq0$. For $\omega<0$ we can simply use that $\Im m_f(\omega,k)=-\Im m_f(-\omega,k)$ and this concludes the proof that the imaginary part does not vanish outside of $\{k=0\}\cup\{\omega=0\}$.

From the previous argument, we see that everything boils down to understanding the behavior of $\Re m_f$ in a neighborhood of $(0,0)$. At this point the maximal value is $\frac{1}2\int_0^\ii r\check{g}(r)\,dr$ and the minimal value is $-\epsilon_g2|\mathbb{S}^{d-1}|$ by definition, hence the result follows.
\end{proof}

We remark that 
$$\Re\,m_f(\omega,k)\underset{\substack{k\to0\\ \omega\to0}}\simeq \frac12\int_{0}^\ii t\check{g}(t)\cos\left(\frac{\omega}{2|k|} t\right)\,dt$$
and therefore we can express
$$-\epsilon_g:=\frac{1}{4|\mathbb{S}^{d-1}|}\min_{a\in\R}\int_{0}^\ii t\check{g}(t)\cos(a t)\,dt.$$

In the three physical cases~\eqref{eq:Fermi-gas-positive-temp}--\eqref{eq:Boltzmann-gas-positive-temp}, the function $f$ satisfies the assumptions of the corollary, and therefore $1+\cL_1$ is invertible with bounded inverse when $w$ satisfies~\eqref{eq:condition_negative_part_bis} and~\eqref{eq:condition_negative_part_bis2}. Numerical computations show that $\epsilon_g$ is always $>0$, but usually smaller than the maximum, by a factor 2 to 10. As an illustration, we display the function $\Re m_f(\omega,k)$ for $T=100$ and $\mu=1$ in Figure~\ref{fig:m_f} below.

\begin{figure}[h]
\centering
\includegraphics[width=8cm]{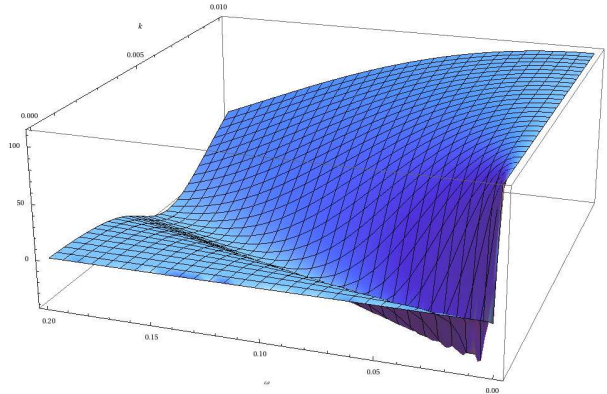}
\caption{Plot of $\Re\, m_f(\omega,k)$ in the fermionic case~\eqref{eq:Fermi-gas-positive-temp} for $d=2$, $T=100$ and $\mu=1$\label{fig:m_f}}
\end{figure}

%%%%%%%%%%%%%%%%%%%%%%%%%%%%%%%%%%%%%%
\subsection{Boundedness of the linear response in $L^p_{t,x}$}
%%%%%%%%%%%%%%%%%%%%%%%%%%%%%%%%%%%%%%

We have studied the boundedness of $\cL_1$ from $L^2_{t,x}$ to $L^2_{t,x}$. This is useful in dimension $d=2$, where the density $\rho_Q$ naturally belongs to $L^2_{t,x}$. However, in other space dimensions, we would like to prove that $\rho_Q$ belongs to $L^{1+2/d}_{t,x}$ and hence, it makes sense to ask whether $\cL_1$ is bounded from $L^p_{t,x}$ to $L^p_{t,x}$. This is the topic of this section. The study of Fourier multipliers acting on $L^p$ is a classical subject in harmonic analysis. We use theorems of Stein and Marcinkiewicz to infer the required boundedness. 

 \begin{proposition}[Boundedness of the linear response on $L^p$]
  Let $w\in L^1(\R^d)$ be such that $|x|^{d+2}w\in L^1(\R^d)$ and such that
  $$\left(\prod_{i\in I}|k_{i}|^2\prod_{j\in J}\partial_{k_{j}}\right)\hat{w}(k)\in L^\ii_k(\R^d),\,\,\forall I\subset\{1,...,d\},\,\,\forall J\subset I.$$
  Let also $h:\R^d\to\R$ be an even function such that
  $$\forall\alpha\in\N^d,\, |\alpha|\le d+3,\quad \int_{\R^d}(1+|k|^{d+4})|\partial^\alpha h(k)|\d{k}<+\ii$$
  and
  $$\left(\prod_{i\in I}\partial_{k_{i}}\right)h\in L^\ii_k(\R^d),\qquad\text{for all $I\subset\{1,...,d\}$}.$$
  Then the Fourier multiplier
  $$\cF_t\big\{\1(t\ge0)\sin(t|k|^2)h(2tk)\big\}$$
  defines a bounded operator from $L^p_{t,x}$ to itself, for every $1<p<\ii$.
 \end{proposition}

The conditions on $h$ are fulfilled if for instance $h$ is a Schwartz function, hence they are fulfilled for our physical examples~\eqref{eq:Fermi-gas-positive-temp}--\eqref{eq:Boltzmann-gas-positive-temp}, where we take $h=\check{g}$.

 \begin{proof}
We define
$$m_1(t,k)=\1(t\ge1)\hat{w}(k)\sin(t|k|^2)h(2tk)$$
and
$$m_2(t,k)=\1(0\leq t\leq 1)\hat{w}(k)\sin(t|k|^2)h(2tk),$$
and use a different criterion for these two multipliers.

To show that $m_1$ defines a bounded operator on $L^p$, we use the criterion of Stein \cite[Thm. 1, II \S2]{Stein-70}. We write $m_1(t,k)=\hat{w}(k)\tilde{m_1}(t,k)$. We first prove estimates on $\tilde{m_1}$, which then imply that $m_1$ defines a bounded Fourier multiplier on $L^p$ by Stein's theorem. Computing the inverse Fourier transform of $\tilde{m_1}$, one has
 $$M_1(t,x):=[\cF^{-1}_k\tilde{m_1}](t,x)=\1(t\ge1)(2\pi)^{-d/2}\int_{\R^d}\sin(t|k|^2)h(2tk)e^{ix\cdot k}\d{k}.$$
 Then, we have
  \begin{equation}\label{eq:nablaxm1NLSgeneral}
      \nabla_x M_1(t,x)=\1(t\ge1)\frac{(2\pi)^{-d/2}}{(2t)^{d+1}}i\int_{\R^d}kh(k)\sin\left(\frac{|k|^2}{4t}\right)e^{i\frac{x\cdot k}{2t}}\,dk.
  \end{equation}
 From this formula, we see that for all $(t,x)$,
 \begin{equation}\label{eq:est-t-nabla-x-M1}
    t^{d+2}|\nabla_x M_1(t,x)|\le C\int_{\R^d}|k|^3|h(k)|\,dk.
 \end{equation}
 Next, let $1\le j\le d$ and notice that
 $$x_j^{d+2}e^{i\frac{x\cdot k}{2t}}=\frac{d^{d+2}}{dk_j^{d+2}}(2t)^{d+2}(-i)^{d+2}e^{i\frac{x\cdot k}{2t}},$$
 and hence by an integration by parts we obtain
 \begin{multline*}
    x_j^{d+2}\nabla_x M_1(t,x)\\
    =\1(t\ge1)(2\pi)^{-d/2}2ti(-i)^{d+2}\int_{\R^d}\frac{d^{d+2}}{dk_j^{d+2}}\left[kh(k)\sin\left(\frac{|k|^2}{4t}\right)\right]e^{i\frac{x\cdot k}{2t}}\,dk.
 \end{multline*}
 When the $k_j$-derivative hits at least once $\sin(|k|^2/4t)$, one gains at least $1/4t$ compensating the $2t$ before the integral; the only term for which we have to prove that it is bounded in $t$ is when all the $k_j$-derivatives hit the term $kh(k)$, which is
 $$\1(t\ge1)(2\pi)^{-d/2}2it(-i)^{d+2}\int_{\R^d}\frac{d^{d+2}}{dk_j^{d+2}}\left[kh(k)\right]\sin\left(\frac{|k|^2}{4t}\right)e^{i\frac{x\cdot k}{2t}}\,dk.$$
 It is also bounded since $|\sin(|k|^2/4t)|\le |k|^2/4t$. We deduce that for all $(t,x)$,
 \begin{equation}\label{eq:est-x-nabla-x-M1}
    |x|^{d+2}|\nabla_xM_1(t,x)|\le C\sup_{\substack{\alpha\in\N^d \\ |\alpha|\le d+2}}\int_{\R^d}(1+|k|^{d+3})|\partial^\alpha h(k)|\,dk.
 \end{equation}
 For the time derivative we use the form
 $$M_1(t,x)=\1(t\ge1)(2\pi)^{-d/2}\int_{\R^d}h(2tk)\sin(t|k|^2)\cos(x\cdot k)\,dk$$
 to infer that for $t\neq1$,
 \begin{align}
    \partial_t M_1(t,x)  =&  2\1(t\ge1)(2\pi)^{-d/2}\int_{\R^d}k\cdot\nabla_kh(2tk)\sin(t|k|^2)\cos(x\cdot k)\d{k}\nonumber\\
     &+\1(t\ge1)(2\pi)^{-d/2}\int_{\R^d}|k|^2h(2tk)\cos(t|k|^2)\cos(x\cdot k)\,dk \nonumber\\
     =& \frac{2\1(t\ge1)}{(2t)^{d+1}}(2\pi)^{-d/2}\int_{\R^d}k\cdot\nabla_kh(k)\sin\left(\frac{|k|^2}{4t}\right)\cos\left(\frac{x\cdot k}{2t}\right)\,dk\nonumber\\
     &+\frac{\1(t\ge1)}{(2t)^{d+2}}(2\pi)^{-d/2}\int_{\R^d}|k|^2h(k)\cos\left(\frac{|k|^2}{4t}\right)\cos\left(\frac{x\cdot k}{2t}\right)\,dk.\label{eq:dtm1NLSgeneral}
 \end{align}
 By the same method as before, we infer
 \begin{equation}\label{eq:est-t-x-nabla-t-M1}
   \|(t,x)\|^{d+2}|\partial_t M_1(t,x)|\le C\sup_{\substack{\alpha\in\N^d \\ |\alpha|\le d+3}}\int_{\R^d}(1+|k|^{d+4})|\partial^\alpha h(k)|\,dk.  
 \end{equation}
Now let us go back to the multiplier $m_1$. We have
 $$\cF_x^{-1}m_1(t,x)=(2\pi)^{d/2}(w\star M_1(t,\cdot))(x),$$
 and hence
 $$\nabla_{t,x}\cF_x^{-1}m_1(t,x)=(2\pi)^{d/2}(w\star\nabla_{t,x}M_1(t,\cdot))(x).$$
 First we have
 $$|t^{d+2}\nabla_{t,x}\cF_x^{-1}m_1(t,x)|\le C\|w\|_{L^1_x}\|t^{d+2}\nabla_{t,x}M_1(t,x)\|_{L^\ii_{t,x}},$$
 which is finite thanks to \eqref{eq:est-t-nabla-x-M1}, \eqref{eq:est-x-nabla-x-M1}, and \eqref{eq:est-t-x-nabla-t-M1}. Next, 
 \begin{multline*}
      |x|^{d+2}|\nabla_{t,x}\cF_x^{-1}m_1(t,x)|\le C\||\cdot|^{d+2}w\|_{L^1_x}\|\nabla_{t,x}M_1(t,x)\|_{L^\ii_{t,x}}\\
      +C\|w\|_{L^1_x}\||x|^{d+2}\nabla_{t,x}M_1(t,x)\|_{L^\ii_{t,x}}.
 \end{multline*}
 The second term is finite also from \eqref{eq:est-t-nabla-x-M1} and \eqref{eq:est-x-nabla-x-M1}, while the first term is finite by the expressions \eqref{eq:nablaxm1NLSgeneral} and \eqref{eq:dtm1NLSgeneral}. As a consequence, we can apply Stein's theorem to $m_1$ and we deduce that the corresponding operator is bounded on $L^p_{t,x}$ for all $1<p<\ii$. 

The multiplier $m_2$ is treated differently. We show that 
 $$m_2\in L^1_t(\R,\cB(L^p_x\to L^p_x)),$$
 which is enough to show that $m_2$ defines a bounded operator on $L^p_{t,x}$. Indeed, for any $\phi\in L^p_{t,x}$, define the Fourier multiplication operator $T_{m_2}$ by
 $$(T_{m_2}\phi)(t,x)=\int_\R \cF_x^{-1}\left[m_2(t-t',\cdot)(\cF_x\phi)(t',\cdot)\right](x)\,dt'.$$
 Then, we have 
 \begin{eqnarray*}
    \|T_{m_2}\phi(t)\|_{L^p_x} & \le & \int_\R\|\cF_x^{-1}[m_2(t-t',\cdot)(\cF_x\phi)(t',\cdot)]\|_{L^p_x}\,dt'\\
     & \le & \int_\R \|m_2(t-t')\|_{\cB(L^p_x\to L^p_x)}\|\phi(t')\|_{L^p_x}\,dt',\\
 \end{eqnarray*}
 and hence
 $$\|T_{m_2}\phi\|_{L^p_{t,x}}\le \|m_2\|_{L^1_t(\R,\cB(L^p_x\to L^p_x))}\|\phi\|_{L^p_{t,x}}.$$
 Hence, let us show that $\|m_2\|_{L^p_x\to L^p_x}\in L^1_t$. We estimate $\|m_2\|_{L^p_x\to L^p_x}$ by the Marcinkiewicz theorem \cite[Cor. 5.2.5]{Grafakos-book}. Namely, we have to show that for all $1\le i_1,\ldots,i_\ell \le d$ all different indices, we have 
 $$k_{i_1}\cdots k_{i_\ell}\partial_{k_{i_1}}\cdots\partial_{k_{i_\ell}}m_2(t,k)\in L^\ii_k,$$
 and if so the Marcinkiewicz theorem tells us that
 $$\|m_2(t)\|_{L^p_x\to L^p_x}\le C\sup_{i_1,\ldots,i_\ell}\|k_{i_1}\cdots k_{i_\ell}\partial_{k_{i_1}}\cdots\partial_{k_{i_\ell}}m_2(t,k)\|_{L^\ii_k}.$$
 A direct computation shows that 
 \begin{multline*}
    |k_{i_1}\cdots k_{i_\ell}\partial_{k_{i_1}}\cdots\partial_{k_{i_\ell}}m_2(t,k)|\\
    \le C\1_{0\le t\le 1}\sum_{I\subset\{i_1,\ldots,i_\ell\}}\sum_{J\subset I}|k_{i_1}|^2\cdots|k_{i_\ell}|^2|\partial_I\hat{w}(k)|(\partial_Jh)(2tk)|,
 \end{multline*}
 where we used the notation $\partial_Jh:=\prod_{j\in J}\partial_{k_j}h$. Hence,
 \begin{multline*}
    \|k_{i_1}\cdots k_{i_\ell}\partial_{k_{i_1}}\cdots\partial_{k_{i_\ell}}m_2(t,k)\|_{L^\ii_k}\\
    \le C\1_{0\le t\le 1}\sup_{I,J\subset\{i_1,\ldots,i_\ell\}}\||k_{i_1}|^2\cdots|k_{i_\ell}|^2|\partial_I\hat{w}(k)|\|_{L^\ii_k}\|\partial_Jh\|_{L^\ii_k},
 \end{multline*}
 which is obviously a $L^1_t$--function. 
\end{proof}

%%%%%%%%%%%%%%%%%%%%%%%%%%%%%%%%%%%%%%%%%%
%%%%%%%%%%%%%%%%%%%%%%%%%%%%%%%%%%%%%%%%%%
\section{Higher order terms}\label{sec:higher-order}
%%%%%%%%%%%%%%%%%%%%%%%%%%%%%%%%%%%%%%%%%%
%%%%%%%%%%%%%%%%%%%%%%%%%%%%%%%%%%%%%%%%%%

In this section, we explain how to treat the higher order terms in \eqref{eq:rhoQ}. We recall the decomposition of the solution for all $t\ge0$:
$$\rho_Q(t)=\rho\left[e^{it\Delta}\cW_{w*\rho_Q}(t)(\gamma_f+Q_0)\cW_{w*\rho_Q}(t)^*e^{-it\Delta}\right]-\rho_{\gamma_f}.$$
We first estimate the terms involving $Q_0$, in dimension 2.

\begin{lemma}\label{lemma:higher-order-Q0}
 Let $Q_0\in\gS^{4/3}(L^2(\R^2))$ and $V\in L^2_{t,x}(\R_+\times\R^2)$. Then, we have the following estimate for all $n,m\ge0$
\begin{multline*}
 \norm{\rho\left[e^{it\Delta}\cW_V^{(n)}(t)Q_0\cW_V^{(m)}(t)^*e^{-it\Delta}\right]}_{L^2_{t,x}(\R_+\times\R^2)}\\
 \le C\norm{Q_0}_{\gS^{4/3}}\frac{C^{n+m}\norm{V}_{L^2_{t,x}}^{n+m}}{(n!)^{\frac 14}(m!)^{\frac 14}},
\end{multline*}
for some $C>0$ independent of $Q_0$, $n$, $m$, and $V$.
\end{lemma}

\begin{proof}
 Defining $\cW^{(0)}_V(t):=1$, for $n,m\ge0$, the density of
 $$e^{it\Delta}\cW^{(n)}_V(t)Q_0\cW^{(m)}_V(t)^*e^{-it\Delta}$$
 is estimated by duality in the following fashion. Let $U\in L^2_{t,x}(\R_+\times\R^2)$. The starting point is the formula
 \begin{multline*}
  \int_0^\ii\int_{\R^2}U(t,x)\rho\left[e^{it\Delta}\cW^{(n)}_V(t)Q_0\cW^{(m)}_V(t)^*e^{-it\Delta}\right](t,x)\,dx\,dt\\
  =\int_0^\ii\tr\left[U(t,x)e^{it\Delta}\cW^{(n)}_V(t)Q_0\cW^{(m)}_V(t)^*e^{-it\Delta}\right]\,dt.
 \end{multline*}
 By cyclicity of the trace, we have
 \begin{multline*}
    \tr\left[U(t,x)e^{it\Delta}\cW^{(n)}_V(t)Q_0\cW^{(m)}_V(t)^*e^{-it\Delta}\right]\\
    =\tr\left[\cW^{(m)}_V(t)^*e^{-it\Delta}U(t,x)e^{it\Delta}\cW^{(n)}_V(t)Q_0\right]
 \end{multline*}
 A straightforward generalization of Theorem \ref{thm:est-Wn} shows that we have 
 \begin{multline*}
    \norm{\int_0^\ii\cW^{(m)}_V(t)^*e^{-it\Delta}U(t,x)e^{it\Delta}\cW^{(n)}_V(t)\,dt}_{\gS^4}\\
    \le \norm{U}_{L^2_{t,x}}\frac{C^n\norm{V}_{L^2_{t,x}}^n}{(n!)^{1/4}}\frac{C^m\norm{V}_{L^2_{t,x}}^m}{(m!)^{1/4}},
 \end{multline*}
 and hence using that $Q_0\in\gS^{4/3}$ and Hölder's inequality, we infer that
 $$\norm{\rho\left[e^{it\Delta}\cW^{(n)}_V(t)Q_0\cW^{(m)}_V(t)^*e^{-it\Delta}\right]}_{L^2_{t,x}}\le\norm{Q_0}_{\gS^{4/3}}\frac{C^n\norm{V}_{L^2_{t,x}}^n}{(n!)^{1/4}}\frac{C^m\norm{V}_{L^2_{t,x}}^m}{(m!)^{1/4}}.$$
 This concludes the proof of the lemma.
\end{proof}

When $d\ge2$, the corresponding result is
\begin{lemma}
 Let $d\ge2$, $Q_0\in\gS^{\frac{d+2}{d+1}}(L^2(\R^d))$, $1<q\le 1+2/d$ and $p$ such that $2/p+d/q=d$. Let $V\in L^{p'}_tL^{q'}_x(\R_+\times\R^d)$. Then, we have the following estimate for any $n,m\ge0$
\begin{multline*}
  \norm{\rho\left[e^{it\Delta}\cW_V^{(n)}(t)Q_0\cW_V^{(m)}(t)^*e^{-it\Delta}\right]}_{L^{p}_tL^{q}_x(\R_+\times\R^d)}\\
  \le C\norm{Q_0}_{\gS^{\frac{d+2}{d+1}}}\frac{C^{n+m}\norm{V}_{L^{p'}_tL^{q'}_x}^{n+m}}{(n!)^{\frac{1}{2q'}}(m!)^{\frac{1}{2q'}}},
\end{multline*}
for some $C>0$ independent of $Q_0$, $n$, $m$, and $V$.
\end{lemma}

The proof follows the same lines as in $d=2$, and relies on the following estimate for any $n,m$ 
\begin{multline*}
    \norm{\int_0^\ii\cW^{(m)}_V(t)^*e^{-it\Delta}U(t,x)e^{it\Delta}\cW^{(n)}_V(t)\,dt}_{\gS^{d+2}}\\
    \le \norm{U}_{L^{p'}_tL^{q'}_x}\frac{C^n\norm{V}_{L^{p'}_tL^{q'}_x}^n}{(n!)^{\frac{1}{2q'}}}\frac{C^m\norm{V}_{L^{p'}_tL^{q'}_x}^m}{(m!)^{\frac{1}{2q'}}}.
 \end{multline*}

We see that the terms involving $Q_0$ can be treated in any dimension, provided that $Q_0$ is in an adequate Schatten space. This is not the case for the terms involving $\gamma_f$, for which we can only deal with the higher orders. 

\begin{lemma}\label{lemma:higher-order-gamma}
 Let $d\ge1$, $g:\R^d\to\R$ such that $\check{g}\in L^1(\R^d)$, $1<q\le 1+2/d$ and $p$ such that $2/p+d/q=d$. Let $V\in L^{p'}_tL^{q'}_x(\R_+\times\R^d)$. Then, for all $n,m$ such that 
 $$n+m+1\ge 2q',$$
 we have 
 \begin{multline*}
  \norm{\rho\left[e^{it\Delta}\cW^{(n)}_V(t)\gamma_f\cW^{(m)}_V(t)^*e^{-it\Delta}\right]}_{L^p_tL^q_x}\\
  \le C\|\check{g}\|_{L^1}\frac{C^n\norm{V}_{L^{p'}_tL^{q'}_x}^n}{(n!)^{\frac{1}{2q'}}}\frac{C^m\norm{V}_{L^{p'}_tL^{q'}_x}^m}{(m!)^{\frac{1}{2q'}}}.
 \end{multline*} 
where $\gamma_f=g(-i\nabla)$.
\end{lemma}

\begin{proof}
 We again argue by duality. Let $U\in L^{p'}_tL^{q'}_x$. Without loss of generality, we can assume that $U,V\ge0$. Then, we evaluate
\begin{multline*}
 \int_0^\ii\tr\left[U(t,x)e^{it\Delta}\cW^{(n)}_V(t)\gamma_f\cW^{(m)}_V(t)^*e^{-it\Delta}\right]\,dt\\
 =(-i)^ni^m\int_0^\ii\,dt\int_{0\le s_1\le\cdots\le s_m\le t}\,ds_1\cdots ds_m\int_{0\le t_1\le\cdots\le t_n\le t}\,dt_1\cdots dt_n\times\\
 \times\tr\left[V(s_1,x-2is_1\nabla)\cdots V(s_m,x-2is_m\nabla)U(t,x-2it\nabla)\times\right.\\
 \left.\times V(t_n,x-2it_n\nabla)\cdots V(t_1,x-2it_1\nabla)\gamma_f\right],
\end{multline*}
 where we used the relation
 $$e^{-it\Delta}W(t,x)e^{it\Delta}=W(t,x-2it\nabla).$$
In the spirit of \cite{FraLewLieSei-13}, we gather the terms using the cyclicity of the trace as 
\begin{multline}\label{eq:bigtrace}
\tr\left[V(s_1,x-2is_1\nabla)\cdots V(s_m,x-2is_m\nabla)U(t,x-2it\nabla)\times\right.\\
 \left.\times V(t_n,x-2it_n\nabla)\cdots V(t_1,x-2it_1\nabla)\gamma_f\right]\\
=\tr\left[V(s_1,x-2is_1\nabla)^{\frac 12}V(s_2,x-2is_2\nabla)^{\frac 12}\cdots\right.\times\\
V(s_m,x-2is_m\nabla)^{\frac 12}U(t,x-2it\nabla)^{\frac 12}U(t,x-2it\nabla)^{\frac 12}V(t_n,x-2it_n\nabla)^{\frac 12}\times\\
\left.\times\cdots V(t_1,x-2it_1\nabla)^{\frac 12}\gamma_fV(s_1,x-2is_1\nabla)^{\frac 12}\right].	 
\end{multline}
The first ingredient to estimate this trace is \cite[Lemma 1]{FraLewLieSei-13}, which states that
\begin{equation}
\norm{\phi_1(\alpha x-i\beta\nabla)\phi_2(\gamma x-i\delta\nabla)}_{\gS^r}\le\frac{\norm{\phi_1}_{L^r(\R^d)}\norm{\phi_2}_{L^r(\R^d)}}{(2\pi)^{\frac dr}|\alpha\delta-\beta\gamma|^{\frac dr}},\quad\forall r\ge2. 
\label{eq:gKSS}
\end{equation}
The second ingredient, to treat the term with $\gamma_f$, is a generalization of this inequality involving $\gamma_f$.

\begin{lemma}\label{lemma:KSS-gamma}
 There exists a constant $C>0$ such that for all $t,s\in\R$ we have 
\begin{equation}
\norm{\phi_1(x+2it\nabla)g(-i\nabla) \phi_2(x+2is\nabla)}_{\gS^r}\le \frac{\|\check{g}\|_{L^1(\R^d)}}{(2\pi)^{\frac{d}{2}}}\frac{\norm{\phi_1}_{L^r(\R^d)}\norm{\phi_2}_{L^r(\R^d)}}{(2\pi)^{\frac{d}{r}}|t-s|^{\frac dr}}
\label{eq:estim_gamma_f} 
\end{equation}
for all $r\geq2$.
\end{lemma}

We remark that~\eqref{eq:estim_gamma_f} reduces to~\eqref{eq:gKSS} when $g=1$ and $\check{g}=(2\pi)^{d/2}\delta_0$.
We postpone the proof of this lemma, and use it to estimate \eqref{eq:bigtrace} in the following way:
\begin{multline*}
 \left|\tr\left[V(s_1,x-2is_1\nabla)\cdots V(s_m,x-2is_m\nabla)U(t,x-2it\nabla)\times\right.\right.\\
 \left.\left.\times V(t_n,x-2it_n\nabla)\cdots V(t_1,x-2it_1\nabla)\gamma_f\right]\right|\\
 \le C\norm{\check{g}}_{L^1}\frac{\norm{V(s_1)}_{L^{q'}}\cdots\norm{V(s_m)}_{L^{q'}}\norm{U(t)}_{L^{q'}}\norm{V(t_n)}_{L^{q'}}\norm{V(t_1)}_{L^{q'}}}{|s_1-t_1|^{\frac{d}{2q'}}\cdots|s_m-t|^{\frac{d}{2q'}}|t-t_n|^{\frac{d}{2q'}}\cdots|t_2-t_1|^{\frac{d}{2q'}}}.
\end{multline*}
Here, we have used the condition $n+m+1\ge2q'$ to ensure that the operator inside the trace is trace-class by Hölder's inequality. From this point the proof is identical to the proof of \cite[Thm. 3]{FraLewLieSei-13}.
\end{proof}

\begin{proof}[Proof of Lemma \ref{lemma:KSS-gamma}]
 The inequality is immediate if $r=\ii$. Hence, by complex interpolation, we only have to prove it for $r=2$. We have 
 \begin{align*}
  \|\phi_1(t,x+2it\nabla) & g(-i\nabla) \phi_2(s,x+2is\nabla)\|_{\gS^2}^2\\
  &= \tr\left[\phi_1(x)^2e^{i(t-s)\Delta}g(-i\nabla)\phi_2(x)^2e^{i(s-t)\Delta}g(-i\nabla)\right]\\
  &= \frac{(2\pi)^{-2d}}{|t-s|^d}\iint \phi_1(x)^2\left|\left(\check{g}*e^{-i\frac{|\cdot|^2}{4(t-s)}}\right)(x-y)\right|^2\phi_2(y)^2\,dx\,dy\\
  &\le \frac{(2\pi)^{-2d}}{|t-s|^d}\norm{\check{g}}_{L^1}^2\norm{\phi_1}_{L^2}^2\norm{\phi_2}_{L^2}^2.
 \end{align*}
\end{proof}

In dimension $d$, we want to prove that $\rho_Q$ belongs to $L^{1+2/d}_{t,x}$, hence we consider $q=1+2/d$ and $q'=1+d/2$. The previous result estimates the terms of order $n+m+1\ge d+2$, that is $n+m\ge d+1$. The case $n+m=1$ corresponds exactly to the linear response studied in the previous section. In dimension $d=2$, we see that we are still lacking the case $n+m=2$, which is what we call the second order. The next section is devoted to this order. We are not able to treat the terms with $1<n+m\le d$ in other dimensions.

%%%%%%%%%%%%%%%%%%%%%%%%%%%%%%%%%%%%%%%%%%
%%%%%%%%%%%%%%%%%%%%%%%%%%%%%%%%%%%%%%%%%%
\section{Second order in 2D}\label{sec:second-order}
%%%%%%%%%%%%%%%%%%%%%%%%%%%%%%%%%%%%%%%%%%
%%%%%%%%%%%%%%%%%%%%%%%%%%%%%%%%%%%%%%%%%%

The study of the linear response is not enough to prove dispersion for the Hartree equation in 2D. We also have to estimate the second order term, that we first compute explicitly in any dimension, and then study only in dimension 2.

%%%%%%%%%%%%%%%%%%%%%%%%
\subsection{Exact computation in any dimension}
 %%%%%%%%%%%%%%%%%%%%%%%%
 
 Define the second order term in the Duhamel expansion of $Q(t)$,
 \begin{multline*}
    Q_2(t):=\\
    (-i)^2\int_0^t\d{s}\int_0^s\d{t_1}e^{i(t-s)\Delta}[V(s),e^{i(s-t_1)\Delta}[V(t_1),\gamma_f]e^{i(t_1-s)\Delta}]e^{i(s-t)\Delta},
 \end{multline*}
 where we again used the notation $V=w*\rho_Q$. We compute explicitly its density. To do so, we let $W\in \cD(\R_+\times\R^d)$ and use the relation
 $$\int_0^\ii\int_{\R^d}W(t,x)\rho_{Q_2}(t,x)\,dx\,dt=\int_0^\ii\tr[W(t)Q_2(t)]\,dt.$$
  For any $(p,q)\in\R^d\times\R^d$ we have
 \begin{multline*}
  \hat{Q_2}(t,p,q)=-\frac{1}{(2\pi)^d}\int_0^tds\int_0^sdt_1\int_{\R^d}dq_1\,e^{i(t-s)(q^2-p^2)}\times\\
  \times\left[\hat{V}(s,p-q_1)e^{i(s-t_1)(q^2-q_1^2)}\hat{V}(t_1,q_1-q)(g(q)-g(q_1))\right.\\
  \left.-\hat{V}(s,q_1-q)e^{i(s-t_1)(q_1^2-p^2)}\hat{V}(t_1,p-q_1)(g(q_1)-g(p))\right].
 \end{multline*}
 Using that
 $$\tr[W(t)Q_2(t)]=\frac{1}{(2\pi)^{d/2}}\int_{\R^d\times\R^d}\hat{W}(t,q-p)\hat{Q_2}(t,p,q)\,dp\,dq,$$
 we arrive at the formula
 \begin{multline*}
    \int_0^\ii\int_{\R^d}W(t,x)\rho_{Q_2}(t,x)\,dx\,dt = \int_0^\ii\int_0^\ii\int_0^\ii\int_{\R^d\times\R^d}\,dt\,ds\,dt_1\,dk\,d\ell\times\\
    \times K^{(2)}(t-s,s-t_1;k,\ell)\hat{W}(t,-k)\hat{\rho_Q}(s,k-\ell)\hat{\rho_Q}(t_1,\ell),
 \end{multline*}
 with
  \begin{multline*}
    K^{(2)}(t,s;k,\ell)\\
    =\1_{t\ge0}\1_{s\ge0}\frac{4\hat{w}(\ell)\hat{w}(k-\ell)}{(2\pi)^{d/2}}\sin(tk\cdot(k-\ell))\sin(\ell\cdot(tk+s\ell))\check{g}(2(tk+s\ell)).
  \end{multline*}

%%%%%%%%%%%%%%%%%%%%%%%%
\subsection{Estimates in 2D}
%%%%%%%%%%%%%%%%%%%%%%%%

\begin{proposition}\label{prop:second-order}
 Assume that $g\in L^1(\R^2)$ is such that $|x|^{a}|\check{g}(x)|\in L^\ii(\R^2)$ for some $a>3$. Assume also that $w$ is such that $(1+|k|^{1/2})|\hat{w}(k)|\in L^\ii(\R^2)$. Then, if $\rho_Q\in L^2_{t,x}(\R\times\R^2)$ we have
 \begin{equation}
  \|\rho_{Q_2}\|_{L^2_{t,x}(\R\times\R^2)}\le C\norm{(1+|\cdot|^2)^{a/2}\check{g}}_{L^\ii}\norm{(1+|\cdot|^{1/2})\hat{w}}_{L^\ii}\|\rho_Q\|_{L^2_{t,x}(\R\times\R^2)}^2,
 \end{equation}
 for some constant $C(g,w)$ only depending on $g$ and $w$. 
\end{proposition}

\begin{proof}
First, we have the estimate
\[
 \left|\int_{\R^3}G(t_1-t_2,t_2-t_3)f_1(t_1)f_2(t_2)f_3(t_3)\,dt_1\,dt_2\,dt_3\right|\le C\|G\|_{L^2L^1}\prod_{i=1}^3\|f_i\|_{L^2},
\]
for any $G$, and hence
\begin{multline*}
 \left|\int_{\R^3}\!\!K^{(2)}(t_1-t_2,t_2-t_3;k,\ell)\hat{W}(t_1,-k)\hat{\rho_Q}(t_2,k-\ell)\hat{\rho_Q}(t_3,\ell)\,dt_1\,dt_2\,dt_3\right|\\
\le \|K^{(2)}(t,s;k,\ell)\|_{L^2_tL^1_s}\|\hat{W}(\cdot,-k)\|_{L^2}\|\hat{\rho_Q}(\cdot,k-\ell)\|_{L^2}\|\hat{\rho_Q}(\cdot,\ell)\|_{L^2}.
\end{multline*}
Let us thus estimate $\|K^{(2)}(t,s;k,\ell)\|_{L^2_tL^1_s}$. To do so, we use $|\sin(tk\cdot(k-\ell))|\le1$, $|\sin(\ell\cdot(tk+s\ell))|\le|\ell||tk+s\ell|$ and get
\begin{multline*}
\|K^{(2)}(t,s;k,\ell)\|_{L^2_tL^1_s}^2\\ \leq \frac{16\hat{w}(\ell)^2\hat{w}(k-\ell)^2}{(2\pi)^{d}}\ell^2\int_{\R}\,dt\left|\int_{\R}\,ds|tk+s\ell||\check{g}(2(tk+s\ell))|\right|^2. 
\end{multline*}
We let 
$$u=\ell s+t\frac{k\cdot\ell}{\ell}\quad \text{and}\quad  v=\sqrt{k^2-\frac{(k\cdot\ell)^2}{\ell^2}}t$$
and notice that 
\[
 |tk+s\ell|=\left(\ell^2 \left(s+t\frac{k\cdot\ell}{\ell^2}\right)^2 +\left(k^2-\frac{(k\cdot\ell)^2}{\ell^2}\right)t^2\right)^{1/2}=\sqrt{u^2+v^2}.
\]
Since $\check{g}$ is a radial function we find that
\begin{multline*}
\ell^2\int_{\R}\,dt\left|\int_{\R}\,ds|tk+s\ell||\check{g}(2(tk+s\ell))|\right|^2\\
=\frac{|\ell|}{\left(k^2\ell^2-(k\cdot\ell)^2\right)^{1/2}}\int_{\R}\,dv\left|\int_{\R}\,du\sqrt{u^2+v^2}|\check{g}(2\sqrt{u^2+v^2})|\right|^2.
\end{multline*}
The double integral on the right is finite under some mild decay assumptions on $\check{g}$, for instance it is finite if $|\check{g}(r)|\leq C(1+r^2)^{-a/2}$, for some $a>3$. Noticing that $\left(k^2\ell^2-(k\cdot\ell)^2\right)^{1/2}=|\det(k,\ell)|$, we thus have
\begin{multline*}
 |\langle W,\rho_{Q_2}\rangle|\le C\norm{(1+|\cdot|^2)^{a/2}\check{g}}_{L^\ii}\int_{\R^{2d}}\,dk\,d\ell\times\\
 \times\frac{\|\hat{W}(\cdot,-k)\|_{L^2}|\hat{w}(k-\ell)|\|\hat{\rho_Q}(\cdot,k-\ell)\|_{L^2}|\ell|^{1/2}|\hat{w}(\ell)|\|\hat{\rho_Q}(\cdot,\ell)\|_{L^2}}{|\det(k,\ell)|^{1/2}}.
\end{multline*}
 We prove the following inequality of Hardy-Littlewood-Sobolev type:

\begin{lemma}\label{lemma:det}
 For any functions $f,g,h$ we have 
 \begin{equation}
  \left|\int_{\R^2\times\R^2}\frac{f(k)g(k-\ell)h(\ell)}{|\det(k,\ell)|^{1/2}}\d{k}\d{\ell}\right|\le C\|f\|_{L^2}\|g\|_{L^2}\|h\|_{L^2}.
 \end{equation}
\end{lemma}

 \begin{proof}
  Since $\det(k,\ell)=k_1\ell_2-k_2\ell_1$, we first fix $k_1\neq0$, $\ell_1\neq0$, $k_1\neq\ell_1$ and estimate
  \begin{multline*}
   \left|\int_{\R^2}\frac{f(k_1,k_2)g(k_1-\ell_1,k_2-\ell_2)h(\ell_1,\ell_2)}{|k_1\ell_2-k_2\ell_1|^{1/2}}\,dk_2\,d\ell_2\right| \\
    \le\left(\int_{\R^2}\frac{|f(k_1,k_2)|^{3/2}|g(k_1-\ell_1,k_2-\ell_2)|^{3/2}}{|k_1\ell_2-k_2\ell_1|^{1/2}}\,dk_2\,d\ell_2\right)^{1/3}\times\\
    \times\left(\int_{\R^2}\frac{|f(k_1,k_2)|^{3/2}|h(\ell_1,\ell_2)|^{3/2}}{|k_1\ell_2-k_2\ell_1|^{1/2}}\,dk_2\,d\ell_2\right)^{1/3}\times\\
    \times\left(\int_{\R^2}\frac{|g(k_1-\ell_1,k_2-\ell_2)|^{3/2}|h(\ell_1,\ell_2)|^{3/2}}{|k_1\ell_2-k_2\ell_1|^{1/2}}\,dk_2\,d\ell_2\right)^{1/3}.
  \end{multline*}
 We then have
 \begin{align*}
  \int_{\R^2}&\frac{|f(k_1,k_2)|^{3/2}|g(k_1-\ell_1,k_2-\ell_2)|^{3/2}}{|k_1\ell_2-k_2\ell_1|^{1/2}}\,dk_2\,d\ell_2 \\
    &= \int_{\R^2}\frac{|f(k_1,k_2)|^{3/2}|g(k_1-\ell_1,\ell_2)|^{3/2}}{|k_2(k_1-\ell_1)-\ell_2k_1|^{1/2}}\,dk_2\,d\ell_2\\
    &= \frac{1}{|k_1||k_1-\ell_1|}\int_{\R^2}\frac{|f(k_1,k_2/(k_1-\ell_1))|^{3/2}|g(k_1-\ell_1,\ell_2/k_1)|^{3/2}}{|k_2-\ell_2|^{1/2}}\,dk_2\,d\ell_2\\
    &\le \frac{C}{|k_1||k_1-\ell_1|}\|f(k_1,\cdot/(k_1-\ell_1))\|_{L^2}^{3/2}\|g(k_1-\ell_1,\cdot/k_1)\|_{L^2}^{3/2}\\
    &\le\frac{C}{|k_1|^{1/4}|k_1-\ell_1|^{1/4}}\|f(k_1,\cdot)\|_{L^2}^{3/2}\|g(k_1-\ell_1,\cdot)\|_{L^2}^{3/2},
 \end{align*}
   and in the same fashion
 \[
  \int_{\R^2}\!\!\frac{|f(k_1,k_2)|^{3/2}|h(\ell_1,\ell_2)|^{3/2}}{|k_1\ell_2-k_2\ell_1|^{1/2}}\,dk_2\,d\ell_2\le\frac{C}{|k_1|^{\frac 14}|\ell_1|^{\frac 14}}\|f(k_1,\cdot)\|_{L^2}^{3/2}\|h(\ell_1,\cdot)\|_{L^2}^{3/2},
 \]
 \begin{multline*}
  \int_{\R^2}\frac{|g(k_1-\ell_1,k_2-\ell_2)|^{3/2}|h(\ell_1,\ell_2)|^{3/2}}{|k_1\ell_2-k_2\ell_1|^{1/2}}\,dk_2\,d\ell_2\\
  \le \frac{C}{|\ell_1|^{1/4}|k_1-\ell_1|^{1/4}}\|g(k_1-\ell_1,\cdot)\|_{L^2}^{3/2}\|h(\ell_1,\cdot)\|_{L^2}^{3/2}.
 \end{multline*}
  As a consequence, we have
  \begin{multline*}
   \left|\int_{\R^2}\frac{f(k_1,k_2)g(k_1-\ell_1,k_2-\ell_2)h(\ell_1,\ell_2)}{|k_1\ell_2-k_2\ell_1|^{1/2}}\,dk_2\,d\ell_2\right|\\
    \le C\frac{\|f(k_1,\cdot)\|_{L^2}\|g(k_1-\ell_1,\cdot)\|_{L^2}\|h(\ell_1,\cdot)\|_{L^2}}{|k_1|^{1/6}|\ell_1|^{1/6}|k_1-\ell_1|^{1/6}}.
  \end{multline*}
We now need a multilinear Hardy-Littlewood-Sobolev-type inequality. Integrating over $(k_1,\ell_1)$ we find that
  \begin{align*}
   \left|\int_{\R^2\times\R^2}\right. & \left.\frac{f(k)g(k-\ell)h(\ell)}{|\det(k,\ell)|^{1/2}}\,dk\,d\ell \right| \\
   &\le C\int_{\R^2}\frac{\|f(k_1,\cdot)\|_{L^2}\|g(k_1-\ell_1,\cdot)\|_{L^2}\|h(\ell_1,\cdot)\|_{L^2}}{|k_1|^{1/6}|\ell_1|^{1/6}|k_1-\ell_1|^{1/6}}\,dk_1\,d\ell_1\\
    &\le C\left(\frac{\|g(k_1-\ell_1,\cdot)\|_{L^2}^{3/2}\|h(\ell_1,\cdot)\|_{L^2}^{3/2}}{|k_1|^{1/2}}\,dk_1\,d\ell_1\right)^{1/3}\times\\
    &\quad\times\left(\frac{\|f(k_1,\cdot)\|_{L^2}^{3/2}\|g(k_1-\ell_1,\cdot)\|_{L^2}^{3/2}}{|\ell_1|^{1/2}}\,dk_1\,d\ell_1\right)^{1/3}\times\\
    &\quad\times\left(\frac{\|f(k_1,\cdot)\|_{L^2}^{3/2}\|h(\ell_1,\cdot)\|_{L^2}^{3/2}}{|k_1-\ell_1|^{1/2}}\,dk_1\,d\ell_1\right)^{1/3}\\
    &\le C\|f\|_{L^2}\|g\|_{L^2}\|h\|_{L^2}
  \end{align*}
where in the last line we have used the 2D Hardy-Littlewood-Sobolev inequality.
 \end{proof}

 From the lemma, we deduce that
 \[
 |\langle W,\rho_{Q_2}\rangle|\le C\norm{(1+|\cdot|^2)^{a/2}\check{g}}_{L^\ii}\norm{(1+|\cdot|^{1/2})\hat{w}}_{L^\ii}\|\rho_Q\|_{L^2_{t,x}}^2,
\]
which ends the proof of the proposition.
\end{proof}
 
%%%%%%%%%%%%%%%%%%%%%%%%%%%%%%%%%%%%%%%%%%
%%%%%%%%%%%%%%%%%%%%%%%%%%%%%%%%%%%%%%%%%%
\section{Proof of the main theorem}\label{sec:proof-thm}
%%%%%%%%%%%%%%%%%%%%%%%%%%%%%%%%%%%%%%%%%%
%%%%%%%%%%%%%%%%%%%%%%%%%%%%%%%%%%%%%%%%%%

\begin{proof}[Proof of Theorem \ref{thm:main}]
Let $T>0$. Assume also that $\norm{Q_0}_{\gS^{4/3}}\le1$. We solve the equation
\begin{align*}
\rho_Q(t)&=\rho\left[e^{it\Delta}\cW_{w*\rho_Q}(t)(\gamma_f+Q_0)\cW_{w*\rho_Q}(t)^*e^{-it\Delta}\right]-\rho_{\gamma_f}\\
&=\rho\left[e^{it\Delta}Q_0e^{-it\Delta}\right]-\cL(\rho_Q)+\cR(\rho_Q)
\end{align*}
by a fixed-point argument. Here $\cL=\cL_1+\cL_2$ where $\cL_1$ was studied in Section~\ref{sec:linear-response} and
$$\cL_2(\rho_Q)=-\rho\left[e^{it\Delta}(\cW^{(1)}_{w*\rho_Q}(t)Q_0+Q_0\cW^{(1)}_{w*\rho_Q}(t)^*)e^{-it\Delta}\right].$$
As explained in Proposition~\ref{prop:linear-response} and in Corollary~\ref{cor:linear-response_defocusing}, under the assumption~\eqref{eq:condition_linear_response} (or~\eqref{eq:condition_negative_part} when $f$ is strictly decreasing), $(1+\cL_1)$ is invertible with bounded inverse on $L^2_{t,x}$. The operator $1+\cL=1+\cL_1+\cL_2$ is invertible with bounded inverse when
$$\norm{\cL_2}<\frac{1}{\norm{(1+\cL_1)^{-1}}}.$$
By Lemma~\ref{lemma:higher-order-Q0}, we have 
$$\norm{\cL_2}\le C\norm{w}_{L^1}\norm{Q_0}_{\gS^{4/3}}$$
and therefore a sufficient condition can be expressed as
$$\norm{Q_0}_{\gS^{4/3}}<\frac{1}{C\norm{w}_{L^1}\norm{(1+\cL_1)^{-1}}}.$$
Then we can write
$$\rho_Q(t)=(1+{\cL})^{-1}\left(\rho\left[e^{it\Delta}Q_0e^{-it\Delta}\right]+\cR(\rho_Q)\right).$$
For any $\phi\in L^2_{t,x}([0,T]\times\R^2)$, define 
$$F(\phi)(t)=\rho\left[e^{it\Delta}Q_0e^{-it\Delta}\right]+\cR(\phi).$$
We apply the Banach fixed-point theorem on the map $(1+\cL)^{-1}F$. To do so, we expand $F$ as
\begin{multline*}
 F(\phi)(t)=\rho\left[e^{it\Delta}Q_0e^{-it\Delta}\right]+\sum_{n+m\ge2}\rho\left[e^{it\Delta}\cW_{w*\phi}(t)Q_0\cW_{w*\phi}(t)^*e^{-it\Delta}\right]\\
 +\sum_{n+m=2}\rho\left[e^{it\Delta}\cW_{w*\phi}^{(n)}(t)\gamma_f\cW_{w*\phi}^{(m)}(t)^*e^{-it\Delta}\right]\\
 +\sum_{n+m\ge3}\rho\left[e^{it\Delta}\cW_{w*\phi}^{(n)}(t)\gamma_f\cW_{w*\phi}^{(m)}(t)^*e^{-it\Delta}\right].
\end{multline*}
By the Strichartz estimate \eqref{eq:Strichartz}, we have 
$$\norm{\rho\left[e^{it\Delta}Q_0e^{-it\Delta}\right]}_{L^2_{t,x}}\le C\norm{Q_0}_{\gS^{4/3}}.$$
By Lemma \ref{lemma:higher-order-Q0}, we have
\begin{multline*}
  \norm{\sum_{n+m\ge2}\rho\left[e^{it\Delta}\cW_{w*\phi}(t)Q_0\cW_{w*\phi}(t)^*e^{-it\Delta}\right]}_{L^2_{t,x}}\\
  \le C\norm{Q_0}_{\gS^{4/3}}\sum_{n+m\ge2}\frac{C^{n+m}\norm{w*\phi}_{L^2_{t,x}}^{n+m}}{(n!)^{\frac 14}(m!)^{\frac 14}}.
\end{multline*}
By Proposition \ref{prop:second-order}, we have
\begin{multline*}
\norm{\sum_{n+m=2}\rho\left[e^{it\Delta}\cW_{w*\phi}^{(n)}(t)\gamma_f\cW_{w*\phi}^{(m)}(t)^*e^{-it\Delta}\right]}_{L^2_{t,x}}\\
\le C\norm{(1+|\cdot|^2)^{a/2})\check{g}}_{L^\ii}\norm{(1+|\cdot|^{1/2})\hat{w}}_{L^\ii}\norm{\phi}_{L^2_{t,x}}^2.
\end{multline*}
Finally, by Lemma \ref{lemma:higher-order-gamma} we have
\begin{multline*}
 \norm{\sum_{n+m\ge3}\rho\left[e^{it\Delta}\cW_{w*\phi}^{(n)}(t)\gamma_f\cW_{w*\phi}^{(m)}(t)^*e^{-it\Delta}\right]}_{L^2_{t,x}}\\
 \le C\|\check{g}\|_{L^1}\sum_{n+m\ge3}\frac{C^{n+m}\norm{w*\phi}_{L^2_{t,x}}^{n+m}}{(n!)^{\frac{1}{4}}(m!)^{\frac{1}{4}}}.
\end{multline*}
We deduce that for all $\phi\in L^2_{t,x}([0,T]\times\R^2)$, we have the estimate
$$\norm{(1+\cL)^{-1}F(\phi)}_{L^2_{t,x}}\le C\norm{(1+\cL)^{-1}}\left(\norm{Q_0}_{\gS^{4/3}}+A(\norm{\phi}_{L^2_{t,x}})\right),$$
where we used the notation
\begin{multline*}
  A(z)=C\sum_{n+m\ge2}\frac{C^{n+m}\left(\norm{w}_{L^1}z\right)^{n+m}}{(n!)^{\frac 14}(m!)^{\frac 14}}\\
  +C\norm{(1+|\cdot|^2)^{a/2})\check{g}}_{L^\ii}\norm{(1+|\cdot|^{1/2})\hat{w}}_{L^\ii}z^2\\
  +C\|\check{g}\|_{L^1}\sum_{n+m\ge3}\frac{C^{n+m}\left(\norm{w}_{L^1}z\right)^{n+m}}{(n!)^{\frac{1}{4}}(m!)^{\frac{1}{4}}}.
\end{multline*}
We have $A(z)=O(z^2)$ as $z\to0$. As a consequence, there exists $C_0,z_0>0$ only depending on $\norm{w}_{L^1}$, $\norm{(1+|\cdot|^2)^{a/2})\check{g}}_{L^\ii}\norm{(1+|\cdot|^{1/2})\hat{w}}_{L^\ii}$, and $\|\check{g}\|_{L^1}$ such that 
$$|A(z)|\le C_0z^2,$$
for all $|z|\le z_0$. Choosing 
$$R=\min\left(z_0,\frac{1}{2C_0\norm{(1+\cL)^{-1}}}\right)$$
and 
$$\norm{Q_0}_{\gS^{4,3}}\le\min\left(1,\frac{R}{2C\norm{(1+\cL)^{-1}}}\right),$$
leads to the estimate
$$\norm{(1+\cL)^{-1}F(\phi)}_{L^2_{t,x}}\le R,$$
for all $\norm{\phi}_{L^2_{t,x}}\le R$, independently of the maximal time $T>0$. Similar estimates show that $F$ is also a contraction on this ball, up to diminishing $R$ if necessary. The Banach fixed point theorem shows that there exists a solution for any time $T>0$, with a uniform estimate with respect to $T$. Having built this solution $\phi_0\in L^2_{t,x}(\R_+\times\R^2)$, we define the operator $\gamma$ as
$$\gamma(t)=e^{it\Delta}\cW_{w*\phi_0}(t)(\gamma_f+Q_0)\cW_{w*\phi_0}(t)^*e^{-it\Delta}.$$
We have $\phi_0=\rho_\gamma-\rho_{\gamma_f}$ by definition. Uniqueness of solutions independently of whether they belong to the ball where we performed the fixed point argument follows from the same arguments as in the proof of \cite[Thm. 5]{LewSab-13a}.

From~\cite[Thm. 3]{FraLewLieSei-13} we know that  $\cW_{w\ast\phi_0}-1\in C^0_t(\R_+,\gS^4)$ and that $\cW_{w\ast\phi_0}-1$ admits a strong limit in $\gS^4$ when $t\to\ii$, which gives that $\gamma-\gamma_f\in C^0(\R_+,\gS^{4})$ and our scattering result~\eqref{eq:scattering}. Next, we remark that 
since $w\in W^{1,1}(\R^2)\subset L^2(\R^2)$, we have $w\ast\phi_0\in L^2_t(L^\ii_x\cap L^2_x)$. From~\cite[Lemma 7]{LewSab-13a} and the fact that $g\in L^2(\R^2)$ (due to~\eqref{eq:derivees_f}), we deduce that $(\cW_{w\ast\phi_0}(t)-1)\gamma_f\in C^0(\R_+,\gS^{2})$. This now shows that $\gamma-\gamma_f\in C^0(\R_+,\gS^{2})$. Of course, we can perform the same procedure for negative times and this finishes the proof of Theorem \ref{thm:main}.
\end{proof}

\bigskip

\noindent\textbf{Acknowledgements.} This work was partially done while the authors were visiting the Centre \'Emile Borel at the Institut Henri Poincar\'e in Paris. The authors acknowledge financial support from the European Research Council under the European Community's Seventh Framework Programme (FP7/2007-2013 Grant Agreement MNIQS 258023), and from the French ministry of research (ANR-10-BLAN-0101).

%%%%%%%%%%%%%%%%%%%%%%%%%%%%%%%%%%%%%%%%%%
%%%%%%%%%%%%%%%%%%%%%%%%%%%%%%%%%%%%%%%%%%

\end{document}